\documentclass[reqno]{amsart}
\usepackage{amssymb}
\usepackage{amsfonts}

\newtheorem{theorem}{Theorem}
\theoremstyle{plain}

\newtheorem{claim}{Claim}

\newtheorem{condition}{Condition}

\newtheorem{corollary}{Corollary}

\newtheorem{definition}{Definition}
\newtheorem{example}{Example}

\newtheorem{lemma}{Lemma}

\newtheorem{proposition}{Proposition}
\newtheorem{remark}{Remark}

\numberwithin{equation}{section}
\input{tcilatex}

\begin{document}
\title[2D and 3D Anisotropic Collapsing Estimates]{Collapsing Estimates and
the Rigorous Derivation of the 2d Cubic Nonlinear Schr\"{o}dinger Equation
with Anisotropic Switchable Quadratic Traps}
\author{Xuwen Chen}
\address{Department of Mathematics\\
University of Maryland\\
College Park, MD 20742}
\email{chenxuwen@math.umd.edu}
\date{01/05/2012}

\begin{abstract}
We consider the 2d and 3d many body Schr\"{o}dinger equations in the
presence of anisotropic switchable quadratic traps. We extend and improve
the collapsing estimates in Klainerman-Machedon \cite{KlainermanAndMachedon}
and Kirkpatrick-Schlein-Staffilani \cite{Kirpatrick}. Together with an
anisotropic version of the generalized lens transform in Carles \cite{Carles}%
, we derive rigorously the cubic NLS with anisotropic switchable quadratic
traps in 2d through a modified Elgart-Erd\"{o}s-Schlein-Yau procedure. For
the 3d case, we establish the uniqueness of the corresponding
Gross-Pitaevskii hierarchy without the assumption of factorized initial data.
\end{abstract}

\maketitle

\section{Introduction}

Bose-Einstein condensation (BEC) is the phenomenon that particles of integer
spin (\textquotedblleft Bosons\textquotedblright ) occupy a macroscopic
quantum state. The first experimental observation of BEC in an interacting
atomic gas occurred in 1995 \cite{Anderson, Davis}. Many similar experiments
were performed later \cite{Philips, Ketterle, Stamper}. In these laboratory
experiments, the particles are initially confined by traps, e.g., the
magnetic fields in \cite{Anderson, Davis}, then the traps are switched in
order to enable observation. To be more precise about the word "switch": in 
\cite{Anderson, Davis} the trap is removed, in \cite{Stamper} the initial
magnetic trap is switched to an optical trap, in \cite{Philips} the trap is
turned off in 2 spatial directions to generate a 2d Bose gas. The dynamic
during the period when the trap is shifting is sophisticated. To model the
evolution in this process, we use a quadratic potential multiplied by a
switch function in each spatial direction for analysis in this paper. This
simplified yet reasonably general model is expected to capture the salient
features of the actual traps: on the one hand the quadratic potential varies
slowly and tends to $\infty $ as $\left\vert x\right\vert \rightarrow \infty 
$; on the other hand, the switch functions describe the space-time
anisotropic properties of the confining potential. In the physics
literature, Lieb, Seiringer and Yngvason remarked in \cite{Lieb1} that the
confining potential is typically $\sim \left\vert x\right\vert ^{2}$ in the
available experiments. Mathematically speaking, the strongest trap we can
deal with in the usual regularity setting of NLS is the quadratic trap since
the work \cite{Yajima} by Yajima and Zhang points out that the ordinary
Strichartz estimates start to fail as the trap exceeds quadratic.

Motivated by the above considerations, we aim to investigate the evolution
of a many-body Boson system during the alteration of the trap. The N-body
wave function $\psi _{N}(\tau ,\overrightarrow{\mathbf{y}_{N}})$ solves the
many body Schr\"{o}dinger equation with anisotropic switchable quadratic
traps: 
\begin{eqnarray}
i\partial _{\tau }\psi _{N} &=&\frac{1}{2}H_{\overrightarrow{\mathbf{y}_{N}}%
}(\tau )\psi _{N}+\frac{1}{N}\sum_{i<j}N^{n\beta }V(N^{\beta }\left( \mathbf{%
y}_{i}-\mathbf{y}_{j}\right) )\psi _{N}  \label{eqn:ManybodySchrodinger} \\
\psi _{N}(0,\overrightarrow{\mathbf{y}_{N}}) &=&\dprod\limits_{j=1}^{N}\phi
_{0}(\mathbf{y}_{j}),  \notag
\end{eqnarray}%
where $\tau \in \mathbb{R}$, $\overrightarrow{\mathbf{y}_{N}}=\left( \mathbf{%
y}_{1},\mathbf{y}_{2},...,\mathbf{y}_{N}\right) \in \mathbb{R}^{nN},$ $V$ is
the interaction between particles, and 
\begin{equation}
H_{\overrightarrow{\mathbf{y}_{N}}}(\tau ):=\sum_{j=1}^{N}H_{\mathbf{y}%
_{j}}(\tau ):=\sum_{j=1}^{N}\left( \sum_{l=1}^{n}\left( -\frac{\partial ^{2}%
}{\partial y_{j,l}^{2}}+\eta _{l}(\tau )y_{j,l}^{2}\right) \right)
\label{Def:HermiteLikeOperator}
\end{equation}%
with the switch functions $\eta _{l}(\tau )$, $l=1,...,n$. Throughout this
paper, we only consider $n=2$ or $3$ and we assume the switch functions $%
\eta _{l}\in C^{1}(\mathbb{R}_{0}^{+}\rightarrow \mathbb{R}_{0}^{+})$
satisfy the following conditions.

\begin{condition}
\label{Condition:EvenExtension}$\dot{\eta}_{l}(0)=0$ i.e. The trap is not at
a switching stage initially.
\end{condition}

\begin{condition}
\label{Condition:FastSwitch}$\dot{\eta}_{l}$ is supported in $[0,T_{0}]$ and 
$T_{0}\sqrt{\sup_{\tau }\left\vert \eta _{l}(\tau )\right\vert }<\frac{\pi }{%
2}.$
\end{condition}

When the trap is fully on, Lieb, Seiringer, Solovej and Yngvason showed that
the ground state of the Hamiltonian exhibits complete BEC in \cite{Lieb2},
provided that the trapping potential $V_{trap}(x)$ satisfies $%
\inf_{\left\vert x\right\vert >R}V_{trap}(x)$ $\rightarrow \infty $ for $%
R\rightarrow \infty $ and the interaction potential is spherically
symmetric. To be more precise, let $\psi _{N,0}$ be the ground state, then 
\begin{equation*}
\gamma _{N,0}^{(1)}\rightarrow \left\vert \phi _{GP}\right\rangle
\left\langle \phi _{GP}\right\vert \text{ as }N\rightarrow \infty
\end{equation*}%
where $\gamma _{N,0}^{(1)}$ is the corresponding one particle marginal
density defined via formula \ref{def:marginal density} and $\phi _{GP}$
minimizes the Gross-Pitaevskii energy functional%
\begin{equation*}
\int \mathbf{(}\left\vert \nabla \phi \right\vert ^{2}+V_{trap}(x)\left\vert
\phi \right\vert ^{2}+4\pi a_{0}\left\vert \phi \right\vert ^{4}\mathbf{)}d%
\mathbf{x.}
\end{equation*}%
Because we are now considering the evolution while the trap is changing, we
start with a BEC state / factorized state in equation \ref%
{eqn:ManybodySchrodinger}.

However, $\psi _{N}$ does not remain a product of one-particle states i.e.%
\begin{equation*}
\psi _{N}(\tau ,\overrightarrow{\mathbf{y}_{N}})\neq
\dprod\limits_{j=1}^{N}\phi (\tau ,\mathbf{y}_{j}),\tau >0
\end{equation*}%
for some one particle state $\phi $. Moreover it is unrealistic to solve the 
$N$-body equation \ref{eqn:ManybodySchrodinger} for large $N$. Thence, to
observe BEC, we have to show mathematically that $\psi _{N}$ is very close
to $\dprod\limits_{j=1}^{N}\phi (\tau ,\mathbf{y}_{j}),$ the mean field
approximation, in an appropriate sense.

Notice that when $\phi \neq \phi ^{\prime }$%
\begin{equation*}
\left\Vert \dprod\limits_{j=1}^{N}\phi (\tau ,\mathbf{y}_{j})-\dprod%
\limits_{j=1}^{N}\phi ^{\prime }(\tau ,\mathbf{y}_{j})\right\Vert
_{2}^{2}\rightarrow 2\text{ }as\text{ }N\rightarrow \infty .
\end{equation*}%
i.e. our desired limit (the BEC state) is not stable against small
perturbations. One way to circumvent this difficulty is to use the concept
of the k-particle marginal density $\gamma _{N}^{(k)}$ associated with $\psi
_{N}\ $defined as%
\begin{equation}
\gamma _{N}^{(k)}(\tau ,\overrightarrow{\mathbf{y}_{k}};\overrightarrow{%
\mathbf{y}_{k}^{\prime }})=\int \psi _{N}(\tau ,\overrightarrow{\mathbf{y}%
_{k}},\overrightarrow{\mathbf{y}_{N-k}})\overline{\psi _{N}(\tau ,%
\overrightarrow{\mathbf{y}_{k}^{\prime }},\overrightarrow{\mathbf{y}_{N-k}}))%
}d\overrightarrow{\mathbf{y}_{N-k}},\text{ }\overrightarrow{\mathbf{y}_{k}},%
\overrightarrow{\mathbf{y}_{k}^{\prime }}\in \mathbb{R}^{nk}.
\label{def:marginal density}
\end{equation}%
Another way is to add a second order correction to the mean field
approximation. See \cite{Chen2ndOrder, GMM1, GMM2}.

In this paper, we take the marginal density approach and establish the
following theorem.

\begin{theorem}
\label{Theorem:BECin2D}Consider the 2d case when $\beta \in \left( 0,\frac{3%
}{4}\right) $. Assume the interaction potential $V$ is nonnegative and
belongs to $L^{1}\cap W^{2,\infty }$ and the switch functions $\eta _{l}$
satisfy Conditions \ref{Condition:EvenExtension} and \ref%
{Condition:FastSwitch}$.$ Moreover, suppose the initial data has bounded
energy per particle%
\begin{equation*}
\sup_{N}\frac{1}{N}\left\langle \psi _{N},H_{N}(\tau )\psi _{N}\right\rangle %
\bigg|_{\tau =0}<\infty .
\end{equation*}%
where the Hamiltonian $H_{N}(\tau )$ is 
\begin{equation*}
H_{N}(\tau )=\frac{1}{2}\sum_{j=1}^{N}\left( \sum_{l=1}^{2}\left( -\frac{%
\partial ^{2}}{\partial y_{j,l}^{2}}+\eta _{l}(\tau )y_{j,l}^{2}\right)
\right) +\frac{1}{N}\sum_{i<j}N^{2\beta }V(N^{\beta }\left( \mathbf{y}_{i}-%
\mathbf{y}_{j}\right) ).
\end{equation*}%
If $\left\{ \gamma _{N}^{(k)}\right\} $ are the marginal densities
associated with $\psi _{N}$, the solution of the N-body Schr\"{o}dinger
equation \ref{eqn:ManybodySchrodinger}, and $\phi $ solves the 2d
Gross-Pitaevskii equation:%
\begin{eqnarray*}
i\partial _{\tau }\phi -\frac{1}{2}H_{\mathbf{y}}(\tau )\phi
&=&b_{0}\left\vert \phi \right\vert ^{2}\phi \\
\phi (0,\mathbf{y}) &=&\phi _{0}(\mathbf{y}),
\end{eqnarray*}%
where $H_{\mathbf{y}}(\tau )$ is the operator inside formula \ref%
{Def:HermiteLikeOperator} and $b_{0}=\int V(x)dx$, then $\forall \tau \in
\lbrack 0,T_{0}]$ and $k\geqslant 1$, we have the convergence: 
\begin{equation*}
\left\Vert \gamma _{N}^{(k)}(\tau ,\overrightarrow{\mathbf{y}_{k}};%
\overrightarrow{\mathbf{y}_{k}^{\prime }})-\dprod\limits_{j=1}^{k}\phi (\tau
,\mathbf{y}_{j})\overline{\phi (\tau ,\mathbf{y}_{j}^{\prime })}\right\Vert
_{L^{2}(d\overrightarrow{\mathbf{y}_{k}}d\overrightarrow{\mathbf{y}%
_{k}^{\prime }})}\rightarrow 0\text{ as }N\rightarrow \infty .
\end{equation*}
\end{theorem}

\begin{example}
We give a simple example to explain the switching process we are considering
here: say 
\begin{eqnarray*}
\eta _{1}(\tau ) &=&C_{1}\text{ when }\tau \in (-\infty ,\frac{1}{2}],\text{ 
}C_{2}\text{ when }\tau \in \lbrack 1,\infty ), \\
\eta _{2}(\tau ) &=&C_{3}\text{ when }\tau \in (-\infty ,\frac{1}{4}],\text{ 
}C_{4}\text{ when }\tau \in \lbrack \frac{3}{2},\infty ).
\end{eqnarray*}%
Then our switching process contains the cases: turning off / on: $C_{2}=0$ / 
$C_{1}=0$ and tuning up / down: $C_{1}\leqslant C_{2}$ / $C_{2}\leqslant
C_{1}.$ As long as $\eta _{1}(\tau )\in C^{1}$ and satisfies Condition \ref%
{Condition:FastSwitch}$,$ $\eta _{1}$ can behave as one likes inside $[\frac{%
1}{2},1]$. Same comment applies to $\eta _{2}$ too. Furthermore, Theorem \ref%
{Theorem:BECin2D} addresses the time intervals $(-\infty ,0]$ and $[\frac{3}{%
2},\infty )$ as well. Since the equation is time translation invariant in
these two intervals, we can use Theorem \ref{Theorem:BECin2D} separately in
each sufficiently small time intervals.
\end{example}

\begin{remark}
Technically, one should interpret Conditions \ref{Condition:EvenExtension}
and \ref{Condition:FastSwitch} in the following way. Due to Condition \ref%
{Condition:EvenExtension}, we have a $C^{1}$ even extension of $\eta _{l}\ $%
i.e. we define $\eta _{l}(\tau )=\eta _{l}(-\tau )$ for $\tau <0.$ The fast
switching condition \ref{Condition:FastSwitch} in fact ensures that $\beta
_{l}$ defined via equation \ref{eqn: Alpha and Beta} is non-zero in $%
[0,T_{0}]$ which is crucial in this paper. See Claim \ref{Claim:Properties
of Alpha and Beta} for the proof.
\end{remark}

\begin{remark}
We assume $\beta \in \left( 0,\frac{3}{4}\right) $ to match
Kirkpatrick-Schlein-Staffilani \cite{Kirpatrick} in which the authors
studied the $\eta _{l}=0$ case. $\beta =0$ will yield a Hartree equation
instead of the cubic NLS.
\end{remark}

The approach with $\gamma _{N}^{(k)}$ has been proven to be successful in
the $\eta _{l}=0$ and $n=3$ case, which corresponds to the evolution after
the removal of the traps, in the fundamental papers \cite{E-E-S-Y1, E-Y1,
E-S-Y1,E-S-Y2,E-S-Y4, E-S-Y5, E-S-Y3} by Elgart, Erd\"{o}s, Schlein, and
Yau. Their program, outlined by Spohn \cite{Spohn}, consists of two
principal parts: on the one hand, they prove that an appropriate limit of
the sequence $\left\{ \gamma _{N}^{(k)}\right\} _{k=1}^{N}$ as $N\rightarrow
\infty $ solves the Gross-Pitaevskii hierarchy 
\begin{equation}
\left( i\partial _{t}+\frac{1}{2}\triangle _{\overrightarrow{\mathbf{x}_{k}}%
}-\frac{1}{2}\triangle _{\overrightarrow{\mathbf{x}_{k}^{\prime }}}\right)
\gamma ^{(k)}=b_{0}\sum_{j=1}^{k}B_{j,k+1}\left( \gamma ^{(k+1)}\right) ,%
\text{ }k=1,...,n,...
\label{equation:Gross-Pitaevskii hiearchy without a trap}
\end{equation}%
where $B_{j,k+1}$ are in formula \ref{formula:B}; on the other hand, they
show that hierarchy \ref{equation:Gross-Pitaevskii hiearchy without a trap}
has a unique solution which is therefore a completely factored state.
However, the uniqueness theory for hierarchy \ref{equation:Gross-Pitaevskii
hiearchy without a trap} is surprisingly delicate due to the fact that it is
a system of infinitely many coupled equations over an unbounded number of
variables. In \cite{KlainermanAndMachedon}, by assuming a space-time bound,
Klainerman and Machedon gave another proof of the uniqueness in \cite{E-S-Y2}
through a collapsing estimate and a board game argument. We call the
space-time estimates of the solution of Schr\"{o}dinger equations restricted
to a subspace of $\mathbb{R}^{n}$ "collapsing estimates". We can interpret
them as local smoothing estimates for which integrating in time results in a
gain of one hidden derivative in the sense of the trace theorem. To be
specific, the collapsing estimate of \cite{KlainermanAndMachedon} reads:
Suppose $u^{(k+1)}$ solves%
\begin{equation*}
\left( i\partial _{t}+\frac{1}{2}\triangle _{\overrightarrow{\mathbf{x}_{k+1}%
}}-\frac{1}{2}\triangle _{\overrightarrow{\mathbf{x}_{k+1}^{\prime }}%
}\right) u^{(k+1)}=0,
\end{equation*}%
there is $C>0,$ independent of $j$, $k$ or $u^{(k+1)}(0,\overrightarrow{%
\mathbf{x}_{k+1}};\overrightarrow{\mathbf{x}_{k+1}^{\prime }})$ s.t. 
\begin{eqnarray}
&&\left\Vert \left( \prod_{j=1}^{k}\left( \nabla _{\mathbf{x}_{j}}\nabla _{%
\mathbf{x}_{j}^{\prime }}\right) \right) u^{(k+1)}(t,\overrightarrow{\mathbf{%
x}_{k}},\mathbf{x}_{1};\overrightarrow{\mathbf{x}_{k}^{\prime }},\mathbf{x}%
_{1})\right\Vert _{L^{2}(\mathbb{R}\times \mathbb{R}^{3k}\times \mathbb{R}%
^{3k})}  \label{estimate:KeyEstimateInKlainermanMachedon} \\
&\leqslant &C\left\Vert \left( \prod_{j=1}^{k+1}\left( \nabla _{\mathbf{x}%
_{j}}\nabla _{\mathbf{x}_{j}^{\prime }}\right) \right) u^{(k+1)}(0,%
\overrightarrow{\mathbf{x}_{k+1}};\overrightarrow{\mathbf{x}_{k+1}^{\prime }}%
)\right\Vert _{L^{2}(\mathbb{R}^{3(k+1)}\times \mathbb{R}^{3(k+1)})}.  \notag
\end{eqnarray}%
Later, the method in Klainerman and Machedon \cite{KlainermanAndMachedon}
was taken up by Kirkpatrick, Schlein, and Staffilani in \cite{Kirpatrick},
where they studied the corresponding problem in 2d, and Chen, Pavlovi\'{c}
and Tzirakis \cite{TChenAndNpGP1, TChenAndNP, TCNPNT}, in which they
considered the 1d and 2d 3-body interaction problem and the general
existence theory of hierarchy \ref{equation:Gross-Pitaevskii hiearchy
without a trap}.

We are interested in the case $\eta _{l}\neq 0$. So we study the
Gross-Pitaevskii hierarchy with anisotropic switchable quadratic traps. That
is a sequence of functions $\left\{ \gamma ^{(k)}(\tau ,\overrightarrow{%
\mathbf{y}_{k}};\overrightarrow{\mathbf{y}_{k}^{\prime }})\right\}
_{k=1}^{\infty }$, where $\tau \in \mathbb{R}$, $\overrightarrow{\mathbf{y}%
_{k}},\overrightarrow{\mathbf{y}_{k}^{\prime }}\in \mathbb{R}^{nk},$ which
are symmetric, in the sense that

\begin{equation*}
\gamma ^{(k)}(\tau ,\overrightarrow{\mathbf{y}_{k}};\overrightarrow{\mathbf{y%
}_{k}^{\prime }})=\overline{\gamma ^{(k)}(\tau ,\overrightarrow{\mathbf{y}%
_{k}^{\prime }};\overrightarrow{\mathbf{y}_{k}})}
\end{equation*}%
and%
\begin{equation*}
\gamma ^{(k)}(\tau ,\mathbf{y}_{\sigma (1)},\mathbf{y}_{\sigma (2)},...,%
\mathbf{y}_{\sigma (k)};\mathbf{y}_{\sigma (1)}^{\prime },\mathbf{y}_{\sigma
(2)}^{\prime },...,\mathbf{y}_{\sigma (k)}^{\prime })=\gamma ^{(k)}(\tau ,%
\mathbf{y}_{1},\mathbf{y}_{2},...,\mathbf{y}_{k};\mathbf{y}_{1}^{\prime },%
\mathbf{y}_{2}^{\prime },...,\mathbf{y}_{k}^{\prime })
\end{equation*}%
for any permutation $\sigma ,$ since we are considering Bosons, and satisfy
the anisotropic switchable quadratic traps Gross-Pitaevskii infinite
hierarchy of equations:

\begin{equation}
\left( i\partial _{\tau }-\frac{1}{2}H_{\overrightarrow{\mathbf{y}_{k}}%
}(\tau )+\frac{1}{2}H_{\overrightarrow{\mathbf{y}_{k}^{\prime }}}(\tau
)\right) \gamma ^{(k)}=b_{0}\sum_{j=1}^{k}B_{j,k+1}\left( \gamma
^{(k+1)}\right) .
\label{equation:Gross-Pitaevskii hiearchy with anisotropic traps}
\end{equation}%
In the above, $B_{j,k+1}=B_{j,k+1}^{1}-B_{j,k+1}^{2}$ are defined as%
\begin{eqnarray}
&&B_{j,k+1}^{1}\left( \gamma ^{(k+1)}\right) (\tau ,\overrightarrow{\mathbf{y%
}_{k}};\overrightarrow{\mathbf{y}_{k}^{\prime }})  \label{formula:B} \\
&=&\int \int \delta (\mathbf{y}_{j}-\mathbf{y}_{k+1})\delta (\mathbf{y}_{j}-%
\mathbf{y}_{k+1}^{\prime })\gamma ^{(k+1)}(\tau ,\overrightarrow{\mathbf{y}%
_{k+1}};\overrightarrow{\mathbf{y}_{k+1}^{\prime }})d\mathbf{y}_{k+1}d%
\mathbf{y}_{k+1}^{\prime }  \notag \\
&&B_{j,k+1}^{2}\left( \gamma ^{(k+1)}\right) (\tau ,\overrightarrow{\mathbf{y%
}_{k}};\overrightarrow{\mathbf{y}_{k}^{\prime }})  \notag \\
&=&\int \int \delta (\mathbf{y}_{j}^{\prime }-\mathbf{y}_{k+1})\delta (%
\mathbf{y}_{j}^{\prime }-\mathbf{y}_{k+1}^{\prime })\gamma ^{(k+1)}(\tau ,%
\overrightarrow{\mathbf{y}_{k+1}};\overrightarrow{\mathbf{y}_{k+1}^{\prime }}%
)d\mathbf{y}_{k+1}d\mathbf{y}_{k+1}^{\prime }.  \notag
\end{eqnarray}%
These Dirac delta functions in $B_{j,k+1}$ are the reason we consider the
collapsing estimates like estimate \ref%
{estimate:KeyEstimateInKlainermanMachedon}.

When the initial data is a BEC / factorized state%
\begin{equation*}
\gamma ^{(k)}(0,\overrightarrow{\mathbf{y}_{k}};\overrightarrow{\mathbf{y}%
_{k}^{\prime }})=\dprod\limits_{j=1}^{k}\phi _{0}(\mathbf{y}_{j})\overline{%
\phi _{0}(\mathbf{y}_{j}^{\prime })},
\end{equation*}%
hierarchy \ref{equation:Gross-Pitaevskii hiearchy with anisotropic traps}
admits one solution%
\begin{equation*}
\gamma ^{(k)}(\tau ,\overrightarrow{\mathbf{y}_{k}};\overrightarrow{\mathbf{y%
}_{k}^{\prime }})=\dprod\limits_{j=1}^{k}\phi (\tau ,\mathbf{y}_{j})%
\overline{\phi (\tau ,\mathbf{y}_{j}^{\prime })},
\end{equation*}%
which is also a BEC state, provided $\phi $ solves the $n-d$
Gross-Pitaevskii equation%
\begin{eqnarray*}
i\partial _{\tau }\phi -\frac{1}{2}H_{\mathbf{y}}(\tau )\phi
&=&b_{0}\left\vert \phi \right\vert ^{2}\phi \\
\phi (0,\mathbf{y}) &=&\phi _{0}(\mathbf{y}).
\end{eqnarray*}%
Hence we would like to have uniqueness theorems of hierarchy \ref%
{equation:Gross-Pitaevskii hiearchy with anisotropic traps}.

\subsection{Main Auxiliary Theorems}

To obtain Theorem \ref{Theorem:BECin2D}, we need the auxiliary theorems in
this subsection which are of independent interest. We show them in 3d as
well. On the one hand, the general idea for the 2d case is derived from the
higher dimensional case. On the other hand, the 2d and 3d cases are
dramatically different when they are viewed in the context of Theorem \ref%
{Theorem:BECin2D}. We will explain this difference between the 2d and 3d
case in Section \ref{Sec:ProofOfBEC}. For the moment, notice that the
uniqueness theorems in 2d and 3d address two different Gross-Pitaevskii
hierarchies which stand for the two sides of the lens transform. Also, we
currently do not have a 3d version of the 2d convergence / Theorem \ref%
{Theorem:BECin2D}. We state our auxiliary theorems regarding different
dimensions separately for comparison.

First, we have the following collapsing estimates which generalizes estimate %
\ref{estimate:KeyEstimateInKlainermanMachedon}.

\begin{theorem}
\label{Theorem:3*nd}(3*n-d optimal collapsing estimate) Let $n=2$ or $3$,
write 
\begin{equation*}
L_{\mathbf{x}}(t)=\sum_{l=1}^{n}a_{l}(t)\frac{\partial ^{2}}{\partial
x_{l}^{2}},
\end{equation*}%
where the $L_{loc}^{1}$ functions $a_{l}$ satisfy 
\begin{equation*}
a_{l}\geqslant c_{0}>0\text{ }a.e.
\end{equation*}%
Assume $u(t,\mathbf{x}_{1},\mathbf{x}_{2},\mathbf{x}_{2}^{\prime })$ solves
the Schr\"{o}dinger equation 
\begin{eqnarray}
iu_{t}+L_{\mathbf{x}_{1}}(t)u+L_{\mathbf{x}_{2}}(t)u\pm L_{\mathbf{x}%
_{2}^{\prime }}(t)u &=&0\text{ in }\mathbb{R}^{3n+1}  \label{eqn:3*3d} \\
u(0,\mathbf{x}_{1},\mathbf{x}_{2},\mathbf{x}_{2}^{\prime }) &=&f(\mathbf{x}%
_{1},\mathbf{x}_{2},\mathbf{x}_{2}^{\prime }),  \notag
\end{eqnarray}%
then%
\begin{equation*}
\int_{\mathbb{R}^{n+1}}\left\vert \left\vert \nabla _{\mathbf{x}}\right\vert
^{\frac{n-1}{2}}u(t,\mathbf{x},\mathbf{x},\mathbf{x})\right\vert ^{2}d%
\mathbf{x}dt\leqslant C\left\Vert \left\vert \nabla _{\mathbf{x}%
_{1}}\right\vert ^{\frac{n-1}{2}}\left\vert \nabla _{\mathbf{x}%
_{2}}\right\vert ^{\frac{n-1}{2}}\left\vert \nabla _{\mathbf{x}_{2}^{\prime
}}\right\vert ^{\frac{n-1}{2}}f\right\Vert _{2}^{2}.
\end{equation*}
\end{theorem}

Theorem \ref{Theorem:3*nd} is a scale invariant estimate when $a_{l}=1$
hence it is optimal. In fact, it holds for all $n\geqslant 2.$ The proof is
different for $n=2$ and $n\geqslant 3$. We name the third spatial variables $%
\mathbf{x}_{2}^{\prime }$ to match the uniqueness theorems. We point out
that Kirkpatrick, Schlein and Staffilani proved the almost optimal result
for the 2d constant coefficient case in \cite{Kirpatrick}. Some other
collapsing estimates were attained in \cite{ChenDie, GM}.

\subsubsection{2d Auxiliary Theorems}

Theorem \ref{Theorem:3*nd}\ is the key to show the following uniqueness
theorem.

\begin{theorem}
\label{Theorem:Uniqueness of 2d unknown GP}(Uniqueness of 2d GP with
time-dependent coefficients) Let $L_{\mathbf{x}_{k}}$ be in Theorem \ref%
{Theorem:3*nd} and $B_{j,k+1}$ be defined via formula \ref{formula:B}. Say $%
\left\{ u^{(k)}(\tau ,\overrightarrow{\mathbf{y}_{k}};\overrightarrow{%
\mathbf{y}_{k}^{\prime }})\right\} _{k=1}^{\infty }$ solves the
Gross-Pitaevskii hierarchy with variable coefficients 
\begin{equation*}
\left( i\partial _{t}+L_{\overrightarrow{\mathbf{x}_{k+1}}}(t)-L_{%
\overrightarrow{\mathbf{x}_{k+1}^{\prime }}}(t)\right)
u^{(k)}=b_{0}\sum_{j=1}^{k}B_{j,k+1}\left( u^{(k+1)}\right) ,
\end{equation*}%
subject to zero initial data and the space-time bound%
\begin{equation*}
\int_{0}^{T}\left\Vert \prod_{j=1}^{k}\left( \left\vert \nabla _{\mathbf{x}%
_{j}}\right\vert ^{\frac{1}{2}}\left\vert \nabla _{\mathbf{x}_{j}^{\prime
}}\right\vert ^{\frac{1}{2}}\right) B_{j,k+1}u^{(k+1)}(t,\mathbf{\cdot };%
\mathbf{\cdot })\right\Vert _{L^{2}(\mathbb{R}^{2k}\times \mathbb{R}%
^{2k})}dt\leqslant C^{k}
\end{equation*}%
for some $C>0$ and all $1\leqslant j\leqslant k.$ Then $\forall k,t\in
\lbrack 0,T]$, 
\begin{equation*}
\left\Vert \prod_{j=1}^{k}\left( \left\vert \nabla _{\mathbf{x}%
_{j}}\right\vert ^{\frac{1}{2}}\left\vert \nabla _{\mathbf{x}_{j}^{\prime
}}\right\vert ^{\frac{1}{2}}\right) u^{(k)}(t,\mathbf{\cdot };\mathbf{\cdot }%
)\right\Vert _{L^{2}(\mathbb{R}^{2k}\times \mathbb{R}^{2k})}=0.
\end{equation*}
\end{theorem}

In contrast to the standard Elgart-Erd\"{o}s-Schlein-Yau program, we do not
need a uniqueness theorem regarding the Gross-Pitaevskii hierarchy with
anisotropic switchable quadratic traps (hierarchy \ref%
{equation:Gross-Pitaevskii hiearchy with anisotropic traps}) to establish
Theorem \ref{Theorem:BECin2D}. It is enough to have Theorem \ref%
{Theorem:Uniqueness of 2d unknown GP} which has no quadratic potential
inside. At a glance, the analysis of the above hierarchy based on the
Laplacian is unrelated to the hierarchy \ref{equation:Gross-Pitaevskii
hiearchy with anisotropic traps} based on a Hermite like operator $H_{%
\mathbf{y}}(\tau ).$ However, Carles' generalized lens transform \cite%
{Carles} links them together. In fact, the generalized lens transform
preserves $L^{2}$ critical NLS and thus the 2d Gross-Pitaevskii hierarchies.
The specific version of the lens transform we need is in Section \ref%
{Sec:TheLensTransform}.

\subsubsection{3d Auxiliary Theorems}

As mentioned before, the uniqueness theorem here addresses a different
hierarchy from Theorem \ref{Theorem:Uniqueness of 2d unknown GP}. Of course
we can prove a 3d version of Theorem \ref{Theorem:Uniqueness of 2d unknown
GP}. However, the disparity between the 2d and 3d case renders such a
theorem of little value because the lens transform does not preserve the 3d
cubic NLS. See Section \ref{Sec:ProofOfBEC} for detail.

We consider the norm%
\begin{equation}
\left\Vert R_{\tau }^{(k)}\gamma ^{(k)}(\tau ,\mathbf{\cdot };\mathbf{\cdot }%
)\right\Vert _{L^{2}(\mathbb{R}^{3k}\times \mathbb{R}^{3k})}
\label{formula:def of norm}
\end{equation}%
in which%
\begin{equation*}
R_{\tau }^{(k)}=\left( \tprod\nolimits_{j=1}^{k}P_{\mathbf{y}_{j}}(\tau )P_{%
\mathbf{y}_{j}^{\prime }}(-\tau )\right)
\end{equation*}%
\begin{equation*}
P_{\mathbf{y}}(\tau )=%
\begin{pmatrix}
i\beta _{1}(\tau )\frac{\partial }{\partial y_{1}}+\dot{\beta}_{1}(\tau
)y_{1} \\ 
i\beta _{2}(\tau )\frac{\partial }{\partial y_{2}}+\dot{\beta}_{2}(\tau
)y_{2} \\ 
i\beta _{3}(\tau )\frac{\partial }{\partial y_{3}}+\dot{\beta}_{3}(\tau
)y_{3}%
\end{pmatrix}%
\end{equation*}%
where $\beta _{l}$ solves 
\begin{equation}
\ddot{\beta}_{l}(\tau )+\eta _{l}(\tau )\beta _{l}(\tau )=0,\beta _{l}(0)=1,%
\dot{\beta}_{l}(0)=0.  \label{eqn:beta}
\end{equation}%
The operator $i\beta _{l}(\tau )\frac{\partial }{\partial y_{l}}+\dot{\beta}%
_{l}(\tau )y_{l}$ was introduced by Carles in \cite{Carles}. Lemma \ref%
{Lemma:Monentum} and relation \ref{equation:naturality} indicate that the
norm \ref{formula:def of norm} is natural. That is because this operator is
in fact the evolution of the momentum operator $-i\nabla $. We will compute
it in the appendix.

Through a specific generalized lens transform (Proposition \ref%
{Proposition:TheLensTransformNeeded}) we produce the collapsing estimate
which is the key estimate to our 3d uniqueness theorem regarding hierarchy %
\ref{equation:Gross-Pitaevskii hiearchy with anisotropic traps} when $n=3$.

\begin{theorem}
\label{Theorem:Collapsing for GP}Let $[s,T]\subset \lbrack 0,T_{0}]$ and $%
\beta _{l}$ be defined through equation \ref{eqn:beta}, assume $\gamma
^{(k+1)}(\tau ,\mathbf{y}_{k+1};\mathbf{y}_{k+1}^{\prime })$ satisfies the
homogeneous equation%
\begin{eqnarray}
\left( i\partial _{\tau }-\frac{1}{2}H_{\overrightarrow{\mathbf{y}_{k+1}}%
}(\tau )+\frac{1}{2}H_{\overrightarrow{\mathbf{y}_{k+1}^{\prime }}}(\tau
)\right) \gamma ^{(k+1)} &=&0
\label{eqn:homogeneous hierarchy with anisotropic traps} \\
\gamma ^{(k+1)}(0,\overrightarrow{\mathbf{y}_{k+1}};\overrightarrow{\mathbf{y%
}_{k+1}^{\prime }}) &=&\gamma _{0}^{(k+1)}(\overrightarrow{\mathbf{y}_{k+1}};%
\overrightarrow{\mathbf{y}_{k+1}^{\prime }}).  \notag
\end{eqnarray}%
Then exists a $C>0$ independent of $\gamma _{0}^{(k+1)},$ $j,$ $k,$ $s,$ and 
$T$ s.t.%
\begin{eqnarray*}
&&\left\Vert R_{\tau }^{(k)}B_{j,k+1}\left( \gamma ^{(k+1)}\right)
\right\Vert _{L^{2}([s,T]\times \mathbb{R}^{3k}\times \mathbb{R}^{3k})}^{2}
\\
&\leqslant &C\left( \inf_{\tau \in \lbrack
0,T_{0}]}\dprod\limits_{l=2}^{3}\beta _{l}^{2}(\tau )\right) ^{-1}\left\Vert
R_{\tau }^{(k+1)}\gamma ^{(k+1)}\right\Vert _{L^{2}(\mathbb{R}%
^{3(k+1)}\times \mathbb{R}^{3(k+1)})}^{2},
\end{eqnarray*}%
where the $\tau $ on the RHS of the above estimate can be chosen freely in $%
[s,T],$
\end{theorem}

From Theorem \ref{Theorem:Collapsing for GP}, it follows

\begin{theorem}
\label{Theorem:Uniqueness of GP}(Uniqueness of 3d GP with anisotropic
switchable quadratic traps) Let $\left\{ \gamma ^{(k)}(\tau ,\overrightarrow{%
\mathbf{y}_{k}};\overrightarrow{\mathbf{y}_{k}^{\prime }})\right\}
_{k=1}^{\infty }$ solve the 3d Gross-Pitaevskii hierarchy with anisotropic
switchable quadratic traps (hierarchy \ref{equation:Gross-Pitaevskii
hiearchy with anisotropic traps} when $n=3$) subject to zero initial data
and the space-time bound%
\begin{equation}
\int_{0}^{T_{0}}\left\Vert R_{\tau }^{(k)}B_{j,k+1}\gamma ^{(k+1)}(\tau ,%
\mathbf{\cdot };\mathbf{\cdot })\right\Vert _{L^{2}(\mathbb{R}^{3k}\times 
\mathbb{R}^{3k})}d\tau \leqslant C^{k}  \label{formula:space-time bound}
\end{equation}%
for some $C>0$ and all $1\leqslant j\leqslant k.$ Then $\forall k,\tau \in
\lbrack 0,T_{0}]$, 
\begin{equation*}
\left\Vert R_{\tau }^{(k)}\gamma ^{(k)}(\tau ,\mathbf{\cdot };\mathbf{\cdot }%
)\right\Vert _{L^{2}(\mathbb{R}^{3k}\times \mathbb{R}^{3k})}=0.
\end{equation*}
\end{theorem}

\begin{remark}
It is currently unknown how to show directly that the limit of $\gamma
_{N}^{(k)}$ in 3d satisfies the space-time bound \ref{formula:space-time
bound}.
\end{remark}

\subsection{Organization of the Paper}

We show Theorem \ref{Theorem:3*nd} for $n=3$ first in Section \ref%
{Sec:ProofOf3*3d}. Utilizing the same scheme, we prove Theorem \ref%
{Theorem:3*nd} for $n=2$ in Section \ref{Sec:ProofOf3*2d}. Compared to \cite%
{KlainermanAndMachedon} which uses the approach in the Klainerman-Machedon
null form paper \cite{KlainermanMachedonNullForm}, the proofs of Theorem \ref%
{Theorem:3*nd} here are closer to Beals and Bezard \cite{BealsAndBezard}
which is a simplification of \cite{KlainermanMachedonNullForm} in the sense
that duality takes the place of convolution with surface measures.

In Section \ref{Sec:TheLensTransform}, we lay down the tools, a generalized
lens transform and its related properties, involved in establishing Theorems %
\ref{Theorem:Collapsing for GP} and \ref{Theorem:Uniqueness of GP} whose
proofs are in Sections \ref{Sec:ProofOfGPCollapsing} and \ref%
{Sec:ProofOfUniqueness}. Theorem \ref{Theorem:Uniqueness of 2d unknown GP}
follows from the same procedure.

In Section \ref{Sec:ProofOfBEC}, we put together the generalized lens
transform, Theorem \ref{Theorem:Uniqueness of 2d unknown GP}, and the
estimates in Kirkpatrick-Schlein-Staffilani \cite{Kirpatrick}\ to establish
Theorem \ref{Theorem:BECin2D}. We also explain the differences between the
2d and 3d cases there.

In the appendix, we present an algebraic explanation of the generalized lens
transform, one of the vital tools in this paper.

\subsection{Acknowledgment}

The author's thanks go to Professor Matei Machedon and Professor Manoussos
G. Grillakis for the discussion related to this work and pointing out to him
the connection between the generalized lens transform and the metaplectic
representation, to Professor R\'{e}mi Carles for sharing the history of the
lens transform with us, to Mr. Kwan-yuet Ho for telling the author about 
\cite{Philips}, and to Miss Victoria Taroudaki for translating the abstract
of the paper into French.

\section{Proof of Theorem \protect\ref{Theorem:3*nd} when $n=3$ / 3*3d
Collapsing Estimate\label{Sec:ProofOf3*3d}}

We will make use of the lemma.

\begin{lemma}
\label{Lemma:MateiLemmaForIntegrals}\cite{KlainermanAndMachedon}Let $\mathbf{%
\xi }\in \mathbb{R}^{3}$ and $P$ be a 2d plane or sphere in $\mathbb{R}^{3}$
with the usual induced surface measure $dS$.

(1) Say $0<a,$ $b<2,$ $a+b>2,$ then%
\begin{equation*}
\int_{P}\frac{dS(\mathbf{\eta })}{\left\vert \mathbf{\xi }-\mathbf{\eta }%
\right\vert ^{a}\left\vert \mathbf{\eta }\right\vert ^{b}}\leqslant \frac{C}{%
\left\vert \mathbf{\xi }\right\vert ^{a+b-2}}.
\end{equation*}

(2) Say $\varepsilon =\frac{1}{10}$, then%
\begin{equation*}
\int_{P}\frac{dS(\mathbf{\eta })}{\left\vert \frac{\mathbf{\xi }}{2}-\mathbf{%
\eta }\right\vert \left\vert \mathbf{\xi }-\mathbf{\eta }\right\vert
^{2-\varepsilon }\left\vert \mathbf{\eta }\right\vert ^{2-\varepsilon }}%
\leqslant \frac{C}{\left\vert \mathbf{\xi }\right\vert ^{3-2\varepsilon }}.
\end{equation*}

Both the constants in the above estimates are independent of $P.$
\end{lemma}

\begin{proof}
See pages 174 - 175 of \cite{KlainermanAndMachedon}.
\end{proof}

By duality, to gain Theorem \ref{Theorem:3*nd} when $n=3$, it suffices to
prove%
\begin{equation*}
\left\vert \int_{\mathbb{R}^{3+1}}\left\vert \nabla _{\mathbf{x}}\right\vert
u(t,\mathbf{x},\mathbf{x},\mathbf{x})h(t,\mathbf{x})d\mathbf{x}dt\right\vert
\leqslant C\left\Vert h\right\Vert _{2}\left\Vert \nabla _{\mathbf{x}%
_{1}}\nabla _{\mathbf{x}_{2}}\nabla _{\mathbf{x}_{2}^{\prime }}f\right\Vert
_{2}.
\end{equation*}%
Let%
\begin{equation*}
A_{t}=%
\begin{pmatrix}
\int_{0}^{t}a_{1}(s)ds & 0 & 0 \\ 
0 & \int_{0}^{t}a_{2}(s)ds & 0 \\ 
0 & 0 & \int_{0}^{t}a_{3}(s)ds%
\end{pmatrix}%
,
\end{equation*}%
then it brings the solution of equation \ref{eqn:3*3d}%
\begin{equation*}
u(t,\mathbf{x}_{1},\mathbf{x}_{2},\mathbf{x}_{2}^{\prime })=\int e^{i(%
\mathbf{\xi }_{1}^{T}A_{t}\mathbf{\xi }_{1}+\mathbf{\xi }_{2}^{T}A_{t}%
\mathbf{\xi }_{2}\pm \left( \mathbf{\xi }_{2}^{\prime }\right) ^{T}A_{t}%
\mathbf{\xi }_{2}^{\prime })}e^{i\mathbf{x}_{1}\mathbf{\xi }_{1}}e^{i\mathbf{%
x}_{2}\mathbf{\xi }_{2}}e^{i\mathbf{x}_{2}^{\prime }\mathbf{\xi }%
_{2}^{\prime }}\hat{f}(\mathbf{\xi }_{1},\mathbf{\xi }_{2},\mathbf{\xi }%
_{2}^{\prime })d\mathbf{\xi }_{1}d\mathbf{\xi }_{2}d\mathbf{\xi }%
_{2}^{\prime }.
\end{equation*}%
Accordingly, the spatial Fourier transform of $\left\vert \nabla _{\mathbf{x}%
}\right\vert u(t,\mathbf{x},\mathbf{x},\mathbf{x})$ is 
\begin{equation*}
\left\vert \mathbf{\xi }_{1}\right\vert \int e^{i(\left( \mathbf{\xi }_{1}-%
\mathbf{\xi }_{2}-\mathbf{\xi }_{2}^{\prime }\right) ^{T}A_{t}\left( \mathbf{%
\xi }_{1}-\mathbf{\xi }_{2}-\mathbf{\xi }_{2}^{\prime }\right) +\mathbf{\xi }%
_{2}^{T}A_{t}\mathbf{\xi }_{2}\pm \left( \mathbf{\xi }_{2}^{\prime }\right)
^{T}A_{t}\mathbf{\xi }_{2}^{\prime })}\hat{f}(\mathbf{\xi }_{1}-\mathbf{\xi }%
_{2}-\mathbf{\xi }_{2}^{\prime },\mathbf{\xi }_{2},\mathbf{\xi }_{2}^{\prime
})d\mathbf{\xi }_{2}d\mathbf{\xi }_{2}^{\prime },
\end{equation*}%
which allows us to compute that%
\begin{eqnarray*}
&&\left\vert \int \left\vert \nabla _{\mathbf{x}}\right\vert u(t,\mathbf{x},%
\mathbf{x},\mathbf{x})h(t,\mathbf{x})d\mathbf{x}dt\right\vert ^{2} \\
&=&\bigg|\int \left\vert \mathbf{\xi }_{1}\right\vert e^{i(\left( \mathbf{%
\xi }_{1}-\mathbf{\xi }_{2}-\mathbf{\xi }_{2}^{\prime }\right)
^{T}A_{t}\left( \mathbf{\xi }_{1}-\mathbf{\xi }_{2}-\mathbf{\xi }%
_{2}^{\prime }\right) +\mathbf{\xi }_{2}^{T}A_{t}\mathbf{\xi }_{2}\pm \left( 
\mathbf{\xi }_{2}^{\prime }\right) ^{T}A_{t}\mathbf{\xi }_{2}^{\prime })}%
\hat{f}(\mathbf{\xi }_{1}-\mathbf{\xi }_{2}-\mathbf{\xi }_{2}^{\prime },%
\mathbf{\xi }_{2},\mathbf{\xi }_{2}^{\prime }) \\
&&\hat{h}(t,\mathbf{\xi }_{1})dtd\mathbf{\xi }_{1}d\mathbf{\xi }_{2}d\mathbf{%
\xi }_{2}^{\prime }\bigg|^{2}\text{ (spatial Fourier transform on }h\text{)}
\\
&=&\bigg|\int \left( \int \left\vert \mathbf{\xi }_{1}\right\vert
e^{i(\left( \mathbf{\xi }_{1}-\mathbf{\xi }_{2}-\mathbf{\xi }_{2}^{\prime
}\right) ^{T}A_{t}\left( \mathbf{\xi }_{1}-\mathbf{\xi }_{2}-\mathbf{\xi }%
_{2}^{\prime }\right) +\mathbf{\xi }_{2}^{T}A_{t}\mathbf{\xi }_{2}\pm \left( 
\mathbf{\xi }_{2}^{\prime }\right) ^{T}A_{t}\mathbf{\xi }_{2}^{\prime })}%
\hat{h}(t,\mathbf{\xi }_{1})dt\right) \\
&&\hat{f}(\mathbf{\xi }_{1}-\mathbf{\xi }_{2}-\mathbf{\xi }_{2}^{\prime },%
\mathbf{\xi }_{2},\mathbf{\xi }_{2}^{\prime })d\mathbf{\xi }_{1}d\mathbf{\xi 
}_{2}d\mathbf{\xi }_{2}^{\prime }\bigg|^{2} \\
&\leqslant &I(h)\left\Vert \nabla _{\mathbf{x}_{1}}\nabla _{\mathbf{x}%
_{2}}\nabla _{\mathbf{x}_{2}^{\prime }}f\right\Vert _{L^{2}}^{2}\text{
(Cauchy-Schwarz)}
\end{eqnarray*}%
where%
\begin{equation*}
I(h)=\int \frac{\left\vert \mathbf{\xi }_{1}\right\vert ^{2}\left\vert \int
e^{i(\left( \mathbf{\xi }_{1}-\mathbf{\xi }_{2}-\mathbf{\xi }_{2}^{\prime
}\right) ^{T}A_{t}\left( \mathbf{\xi }_{1}-\mathbf{\xi }_{2}-\mathbf{\xi }%
_{2}^{\prime }\right) +\mathbf{\xi }_{2}^{T}A_{t}\mathbf{\xi }_{2}\pm \left( 
\mathbf{\xi }_{2}^{\prime }\right) ^{T}A_{t}\mathbf{\xi }_{2}^{\prime })}%
\hat{h}(t,\mathbf{\xi }_{1})dt\right\vert ^{2}}{\left\vert \mathbf{\xi }_{1}-%
\mathbf{\xi }_{2}-\mathbf{\xi }_{2}^{\prime }\right\vert ^{2}\left\vert 
\mathbf{\xi }_{2}\right\vert ^{2}\left\vert \mathbf{\xi }_{2}^{\prime
}\right\vert ^{2}}d\mathbf{\xi }_{1}d\mathbf{\xi }_{2}d\mathbf{\xi }%
_{2}^{\prime }.
\end{equation*}%
So the target of the rest of this section is to show 
\begin{equation*}
I(h)\leqslant C\left\Vert h\right\Vert _{L^{2}}^{2}.
\end{equation*}%
Noticing that the integral $I(h)$ is symmetric in $\left\vert \mathbf{\xi }%
_{1}-\mathbf{\xi }_{2}-\mathbf{\xi }_{2}^{\prime }\right\vert $ and $%
\left\vert \mathbf{\xi }_{2}\right\vert ,$ we deal with the region: $%
\left\vert \mathbf{\xi }_{1}-\mathbf{\xi }_{2}-\mathbf{\xi }_{2}^{\prime
}\right\vert >\left\vert \mathbf{\xi }_{2}\right\vert $ only. We separate
this region into two parts, Cases I and II.

When the "$\pm $" in equation \ref{eqn:3*3d} is $"+"$, Case I is sufficient.
To show the estimate for $"-",$ we need both Cases I and II.

Away from $\left\vert \mathbf{\xi }_{1}-\mathbf{\xi }_{2}-\mathbf{\xi }%
_{2}^{\prime }\right\vert >\left\vert \mathbf{\xi }_{2}\right\vert $, there
are other restrictions on the integration regions in Cases I and II. We
state the restrictions in the beginning of both Cases I and II. Due to the
limited space near "$\int $", we omit the actual region. Please keep this in
mind during reading.

\subsection{Case I: $I(h)$ restricted to the region $\left\vert \mathbf{%
\protect\xi }_{2}^{\prime }\right\vert <\left\vert \mathbf{\protect\xi }%
_{2}\right\vert $ with integration order $d\mathbf{\protect\xi }_{2}$ prior
to $d\mathbf{\protect\xi }_{2}^{\prime }$}

Write the phase function of the $dt$ integral inside $I(h)$ as%
\begin{eqnarray*}
&&\left( \mathbf{\xi }_{1}-\mathbf{\xi }_{2}-\mathbf{\xi }_{2}^{\prime
}\right) ^{T}A_{t}\left( \mathbf{\xi }_{1}-\mathbf{\xi }_{2}-\mathbf{\xi }%
_{2}^{\prime }\right) +\mathbf{\xi }_{2}^{T}A_{t}\mathbf{\xi }_{2}\pm \left( 
\mathbf{\xi }_{2}^{\prime }\right) ^{T}A_{t}\mathbf{\xi }_{2}^{\prime } \\
&=&\frac{\left( \mathbf{\xi }_{1}-\mathbf{\xi }_{2}^{\prime }\right)
^{T}A_{t}\left( \mathbf{\xi }_{1}-\mathbf{\xi }_{2}^{\prime }\right) }{2}%
+2\left( \mathbf{\xi }_{2}-\frac{\mathbf{\xi }_{1}-\mathbf{\xi }_{2}^{\prime
}}{2}\right) ^{T}A_{t}\left( \mathbf{\xi }_{2}-\frac{\mathbf{\xi }_{1}-%
\mathbf{\xi }_{2}^{\prime }}{2}\right) \pm \left( \mathbf{\xi }_{2}^{\prime
}\right) ^{T}A_{t}\mathbf{\xi }_{2}^{\prime }.
\end{eqnarray*}%
The change of variable%
\begin{equation}
\mathbf{\xi }_{2,new}=\mathbf{\xi }_{2,old}-\frac{\mathbf{\xi }_{1}-\mathbf{%
\xi }_{2}^{\prime }}{2}  \label{formula:change of variable 3*3d sphere}
\end{equation}%
leads to%
\begin{eqnarray*}
I(h) &=&\int \frac{\left\vert \mathbf{\xi }_{1}\right\vert ^{2}\left\vert
\int e^{i(\frac{\left( \mathbf{\xi }_{1}-\mathbf{\xi }_{2}^{\prime }\right)
^{T}A_{t}\left( \mathbf{\xi }_{1}-\mathbf{\xi }_{2}^{\prime }\right) }{2}+2%
\mathbf{\xi }_{2}^{T}A_{t}\mathbf{\xi }_{2}\pm \left( \mathbf{\xi }%
_{2}^{\prime }\right) ^{T}A_{t}\mathbf{\xi }_{2}^{\prime })}\hat{h}(t,%
\mathbf{\xi }_{1})dt\right\vert ^{2}}{\left\vert \mathbf{\xi }_{2}-\frac{%
\mathbf{\xi }_{1}-\mathbf{\xi }_{2}^{\prime }}{2}\right\vert ^{2}\left\vert 
\mathbf{\xi }_{2}+\frac{\mathbf{\xi }_{1}\mathbf{-\xi }_{2}^{\prime }}{2}%
\right\vert ^{2}\left\vert \mathbf{\xi }_{2}^{\prime }\right\vert ^{2}}d\xi
_{1}d\xi _{2}d\xi _{2}^{\prime } \\
&=&\int \frac{\left\vert \mathbf{\xi }_{1}\right\vert ^{2}}{\left\vert 
\mathbf{\xi }_{2}-\frac{\mathbf{\xi }_{1}-\mathbf{\xi }_{2}^{\prime }}{2}%
\right\vert ^{2}\left\vert \mathbf{\xi }_{2}+\frac{\mathbf{\xi }_{1}-\mathbf{%
\xi }_{2}^{\prime }}{2}\right\vert ^{2}\left\vert \mathbf{\xi }_{2}^{\prime
}\right\vert ^{2}}e^{i(2\frac{\left( \mathbf{\xi }_{1}-\mathbf{\xi }%
_{2}^{\prime }\right) ^{T}A_{t}\left( \mathbf{\xi }_{1}-\mathbf{\xi }%
_{2}^{\prime }\right) }{2}+2\mathbf{\xi }_{2}^{T}A_{t}\mathbf{\xi }_{2}\pm
\left( \mathbf{\xi }_{2}^{\prime }\right) ^{T}A_{t}\mathbf{\xi }_{2}^{\prime
})} \\
&&e^{-i(\frac{\left( \mathbf{\xi }_{1}\mathbf{-\xi }_{2}^{\prime }\right)
^{T}A_{t^{\prime }}\left( \mathbf{\xi }_{1}\mathbf{-\xi }_{2}^{\prime
}\right) }{2}+2\mathbf{\xi }_{2}^{T}A_{t^{\prime }}\mathbf{\xi }_{2}\pm
\left( \mathbf{\xi }_{2}^{\prime }\right) ^{T}A_{t^{\prime }}\mathbf{\xi }%
_{2}^{\prime })}\hat{h}(t,\mathbf{\xi }_{1})\overline{\hat{h}(t^{\prime },%
\mathbf{\xi }_{1})}dtdt^{\prime }d\mathbf{\xi }_{1}d\mathbf{\xi }_{2}d%
\mathbf{\xi }_{2}^{\prime } \\
&=&\int d\mathbf{\xi }_{1}\int J(\overline{\hat{h}})(t,\mathbf{\xi }_{1})%
\hat{h}(t,\mathbf{\xi }_{1})dt
\end{eqnarray*}%
where 
\begin{eqnarray*}
J(\overline{\hat{h}})(t,\mathbf{\xi }_{1}) &=&\int \frac{\left\vert \mathbf{%
\xi }_{1}\right\vert ^{2}e^{i2\mathbf{\xi }_{2}^{T}A_{t}\mathbf{\xi }%
_{2}}e^{-i2\mathbf{\xi }_{2}^{T}A_{t^{\prime }}\mathbf{\xi }_{2}}}{%
\left\vert \mathbf{\xi }_{2}-\frac{\mathbf{\xi }_{1}-\mathbf{\xi }%
_{2}^{\prime }}{2}\right\vert ^{2}\left\vert \mathbf{\xi }_{2}+\frac{\mathbf{%
\xi }_{1}-\mathbf{\xi }_{2}^{\prime }}{2}\right\vert ^{2}\left\vert \mathbf{%
\xi }_{2}^{\prime }\right\vert ^{2}} \\
&&e^{i(\frac{\left( \mathbf{\xi }_{1}-\mathbf{\xi }_{2}^{\prime }\right)
^{T}\left( A_{t}-A_{t^{\prime }}\right) \left( \mathbf{\xi }_{1}-\mathbf{\xi 
}_{2}^{\prime }\right) }{2}\pm \left( \mathbf{\xi }_{2}^{\prime }\right)
^{T}\left( A_{t}-A_{t^{\prime }}\right) \mathbf{\xi }_{2}^{\prime })}%
\overline{\hat{h}(t^{\prime },\mathbf{\xi }_{1})}dt^{\prime }d\mathbf{\xi }%
_{2}d\mathbf{\xi }_{2}^{\prime }.
\end{eqnarray*}

Assume for the moment that 
\begin{equation*}
\int \left\vert J(\overline{\hat{h}})(t,\mathbf{\xi }_{1})\right\vert
^{2}dt\leqslant C\left\Vert \hat{h}(\cdot ,\mathbf{\xi }_{1})\right\Vert
_{L_{t}^{2}}^{2}
\end{equation*}%
with $C$ independent of $h$ or $\mathbf{\xi }_{1}$, then%
\begin{equation*}
I(h)\leqslant C\int d\mathbf{\xi }_{1}\left\Vert \hat{h}(\cdot ,\mathbf{\xi }%
_{1})\right\Vert _{L_{t}^{2}}^{2}.
\end{equation*}

Hence we end Case I by this proposition.

\begin{proposition}
\begin{equation*}
\int \left\vert J(f)(t,\mathbf{\xi }_{1})\right\vert ^{2}dt\leqslant
C\left\Vert f(\cdot ,\mathbf{\xi }_{1})\right\Vert _{L_{t}^{2}}^{2}
\end{equation*}%
where $C$ is independent of $f$ or $\mathbf{\xi }_{1}.$
\end{proposition}

\begin{remark}
To avoid confusing notation in the proof of the proposition, we use $%
f(t^{\prime },\mathbf{\xi }_{1})$ to replace $\overline{\hat{h}(t^{\prime },%
\mathbf{\xi }_{1})}.$
\end{remark}

\begin{proof}
Again, by duality, we just need to prove%
\begin{equation*}
\left\vert \int J(f)(t,\mathbf{\xi }_{1})\overline{g(t)}dt\right\vert
\leqslant C\left\Vert f(\cdot ,\mathbf{\xi }_{1})\right\Vert
_{L_{t}^{2}}\left\Vert g\right\Vert _{L_{t}^{2}}.
\end{equation*}%
For convenience, let%
\begin{equation*}
\phi (t,\mathbf{\xi }_{1},\mathbf{\xi }_{2}^{\prime })=\frac{\left( \mathbf{%
\xi }_{1}-\mathbf{\xi }_{2}^{\prime }\right) ^{T}A_{t}\left( \mathbf{\xi }%
_{1}-\mathbf{\xi }_{2}^{\prime }\right) }{2}\pm \left( \mathbf{\xi }%
_{2}^{\prime }\right) ^{T}A_{t}\mathbf{\xi }_{2}^{\prime }.
\end{equation*}%
Then%
\begin{eqnarray*}
&&\left\vert \int J(f)(t,\mathbf{\xi }_{1})\overline{g(t)}dt\right\vert \\
&=&\left\vert \int \frac{\left\vert \mathbf{\xi }_{1}\right\vert ^{2}e^{i2%
\mathbf{\xi }_{2}^{T}A_{t}\mathbf{\xi }_{2}}e^{-i2\mathbf{\xi }%
_{2}^{T}A_{t^{\prime }}\mathbf{\xi }_{2}}}{\left\vert \mathbf{\xi }_{2}-%
\frac{\mathbf{\xi }_{1}-\mathbf{\xi }_{2}^{\prime }}{2}\right\vert
^{2}\left\vert \mathbf{\xi }_{2}+\frac{\mathbf{\xi }_{1}-\mathbf{\xi }%
_{2}^{\prime }}{2}\right\vert ^{2}\left\vert \mathbf{\xi }_{2}^{\prime
}\right\vert ^{2}}\left( e^{-i\phi (t^{\prime },\mathbf{\xi }_{1},\mathbf{%
\xi }_{2}^{\prime })}f(t^{\prime },\mathbf{\xi }_{1})\right) \left( 
\overline{e^{-i\phi (t,\mathbf{\xi }_{1},\mathbf{\xi }_{2}^{\prime })}g(t)}%
\right) dtdt^{\prime }d\mathbf{\xi }_{2}d\mathbf{\xi }_{2}^{\prime
}\right\vert \\
&=&\left\vert \int \frac{\left( \int e^{2i\mathbf{\xi }_{2}^{T}A_{t}\mathbf{%
\xi }_{2}}\left( \overline{e^{-i\phi (t,\mathbf{\xi }_{1},\mathbf{\xi }%
_{2}^{\prime })}g(t)}\right) dt\right) \left( \int e^{-2i\mathbf{\xi }%
_{2}^{T}A_{t^{\prime }}\mathbf{\xi }_{2}}\left( e^{-i\phi (t^{\prime },%
\mathbf{\xi }_{1},\mathbf{\xi }_{2}^{\prime })}f(t^{\prime },\mathbf{\xi }%
_{1})\right) dt^{\prime }\right) \left\vert \mathbf{\xi }_{1}\right\vert
^{2}d\mathbf{\xi }_{2}d\mathbf{\xi }_{2}^{\prime }}{\left\vert \mathbf{\xi }%
_{2}-\frac{\mathbf{\xi }_{1}-\mathbf{\xi }_{2}^{\prime }}{2}\right\vert
^{2}\left\vert \mathbf{\xi }_{2}+\frac{\mathbf{\xi }_{1}-\mathbf{\xi }%
_{2}^{\prime }}{2}\right\vert ^{2}\left\vert \mathbf{\xi }_{2}^{\prime
}\right\vert ^{2}}\right\vert \\
&\leqslant &\int \frac{\left\vert \mathbf{\xi }_{1}\right\vert ^{2}d\mathbf{%
\xi }_{2}^{\prime }}{\left\vert \mathbf{\xi }_{2}^{\prime }\right\vert ^{2}}%
\int \frac{\left\vert \int e^{2i\mathbf{\xi }_{2}^{T}A_{t}\mathbf{\xi }%
_{2}}\left( \overline{e^{-i\phi (t,\mathbf{\xi }_{1},\mathbf{\xi }%
_{2}^{\prime })}g(t)}\right) dt\right\vert \left\vert \int e^{-2i\mathbf{\xi 
}_{2}^{T}A_{t^{\prime }}\mathbf{\xi }_{2}}\left( e^{-i\phi (t^{\prime },%
\mathbf{\xi }_{1},\mathbf{\xi }_{2}^{\prime })}f(t^{\prime },\mathbf{\xi }%
_{1})\right) dt^{\prime }\right\vert }{\left\vert \mathbf{\xi }_{2}-\frac{%
\mathbf{\xi }_{1}-\mathbf{\xi }_{2}^{\prime }}{2}\right\vert ^{2}\left\vert 
\mathbf{\xi }_{2}+\frac{\mathbf{\xi }_{1}-\mathbf{\xi }_{2}^{\prime }}{2}%
\right\vert ^{2}}d\mathbf{\xi }_{2}
\end{eqnarray*}%
To deal with the $dt$ and $dt^{\prime }$ integrals, for every fixed $\mathbf{%
\xi }_{2}$, let%
\begin{equation*}
u(t)=2\frac{\mathbf{\xi }_{2}^{T}A_{t}\mathbf{\xi }_{2}}{\left\vert \mathbf{%
\xi }_{2}\right\vert ^{2}}
\end{equation*}%
then%
\begin{equation*}
\frac{du}{dt}=2\frac{a_{1}(t)\xi _{2,1}^{2}+a_{2}(t)\xi
_{2,2}^{2}+a_{3}(t)\xi _{2,3}^{2}}{\left\vert \mathbf{\xi }_{2}\right\vert
^{2}}\geqslant 2c_{0}>0
\end{equation*}%
which provides a well-defined inverse $t(u)$.

Consequently, the integral%
\begin{equation*}
\int e^{2i\mathbf{\xi }_{2}^{T}A_{t}\mathbf{\xi }_{2}}\left( \overline{%
e^{-i\phi (t,\mathbf{\xi }_{1},\mathbf{\xi }_{2}^{\prime })}g(t)}\right) dt=%
\overline{\int e^{-iu\left\vert \mathbf{\xi }_{2}\right\vert ^{2}}\left(
e^{-i\phi (t(u),\mathbf{\xi }_{1},\mathbf{\xi }_{2}^{\prime
})}g(t(u))\left\vert \frac{dt}{du}\right\vert \right) du},
\end{equation*}%
is indeed the Fourier transform of%
\begin{equation*}
G(u)=e^{-i\phi (t(u),\mathbf{\xi }_{1},\mathbf{\xi }_{2}^{\prime
})}g(t(u))\left\vert \frac{dt}{du}\right\vert .
\end{equation*}%
This is well-defined since%
\begin{equation*}
\int_{\mathbb{R}}\left\vert G(u)\right\vert ^{2}du=\int_{\mathbb{R}%
}\left\vert \overline{e^{-i\phi (t(u),\mathbf{\xi }_{1},\mathbf{\xi }%
_{2})}g(t(u))}\left\vert \frac{dt}{du}\right\vert \right\vert ^{2}du=\int_{%
\mathbb{R}}\left\vert g(t)\right\vert ^{2}\left\vert \frac{dt}{du}%
\right\vert dt\leqslant \frac{1}{2c_{0}}\left\Vert g(\cdot )\right\Vert
_{L_{t}^{2}}^{2}.
\end{equation*}%
Hence%
\begin{eqnarray*}
&&\left\vert \int J(f)(t,\mathbf{\xi }_{1})\overline{g(t)}dt\right\vert \\
&\leqslant &\int \frac{\left\vert \mathbf{\xi }_{1}\right\vert ^{2}d\mathbf{%
\xi }_{2}^{\prime }}{\left\vert \mathbf{\xi }_{2}^{\prime }\right\vert ^{2}}%
\int \frac{\left\vert \int e^{2i\mathbf{\xi }_{2}^{T}A_{t}\mathbf{\xi }%
_{2}}\left( \overline{e^{-i\phi (t,\mathbf{\xi }_{1},\mathbf{\xi }_{2})}g(t)}%
\right) dt\right\vert \left\vert \int e^{-2i\mathbf{\xi }_{2}^{T}A_{t^{%
\prime }}\mathbf{\xi }_{2}}\left( e^{-i\phi (t^{\prime },\mathbf{\xi }_{1},%
\mathbf{\xi }_{2})}f(t^{\prime },\mathbf{\xi }_{1})\right) dt^{\prime
}\right\vert }{\left\vert \mathbf{\xi }_{2}-\frac{\mathbf{\xi }_{1}-\mathbf{%
\xi }_{2}^{\prime }}{2}\right\vert ^{2}\left\vert \mathbf{\xi }_{2}+\frac{%
\mathbf{\xi }_{1}-\mathbf{\xi }_{2}^{\prime }}{2}\right\vert ^{2}}d\mathbf{%
\xi }_{2} \\
&=&\int \frac{\left\vert \mathbf{\xi }_{1}\right\vert ^{2}d\mathbf{\xi }%
_{2}^{\prime }}{\left\vert \mathbf{\xi }_{2}^{\prime }\right\vert ^{2}}\int 
\frac{\left\vert \overline{\hat{G}(\left\vert \mathbf{\xi }_{2}\right\vert
^{2})}\hat{F}(\left\vert \mathbf{\xi }_{2}\right\vert ^{2},\mathbf{\xi }%
_{1})\right\vert }{\left\vert \mathbf{\xi }_{2}-\frac{\mathbf{\xi }_{1}-%
\mathbf{\xi }_{2}^{\prime }}{2}\right\vert ^{2}\left\vert \mathbf{\xi }_{2}+%
\frac{\mathbf{\xi }_{1}-\mathbf{\xi }_{2}^{\prime }}{2}\right\vert ^{2}}d%
\mathbf{\xi }_{2} \\
&=&\int \frac{\left\vert \mathbf{\xi }_{1}\right\vert ^{2}d\mathbf{\xi }%
_{2}^{\prime }}{\left\vert \mathbf{\xi }_{2}^{\prime }\right\vert ^{2}}\int 
\frac{\left\vert \hat{F}(\rho ^{2},\mathbf{\xi }_{1})\overline{\hat{G}(\rho
^{2})}\right\vert }{\left\vert \mathbf{\xi }_{2}-\frac{\mathbf{\xi }_{1}-%
\mathbf{\xi }_{2}^{\prime }}{2}\right\vert ^{2}\left\vert \mathbf{\xi }_{2}+%
\frac{\mathbf{\xi }_{1}-\mathbf{\xi }_{2}^{\prime }}{2}\right\vert ^{2}}\rho
^{2}d\rho d\mathbf{\sigma }\text{ }\left( \text{spherical coordinate in }%
\mathbf{\xi }_{2}\right) \\
&\leqslant &\int \frac{\left\vert \mathbf{\xi }_{1}\right\vert ^{2}d\mathbf{%
\xi }_{2}^{\prime }}{\left\vert \mathbf{\xi }_{2}^{\prime }\right\vert ^{2}}%
\sup_{\rho }\left( \int \frac{\rho ^{2}d\mathbf{\sigma }}{\rho \left\vert 
\mathbf{\xi }_{2}-\frac{\mathbf{\xi }_{1}-\mathbf{\xi }_{2}^{\prime }}{2}%
\right\vert ^{2}\left\vert \mathbf{\xi }_{2}+\frac{\mathbf{\xi }_{1}-\mathbf{%
\xi }_{2}^{\prime }}{2}\right\vert ^{2}}\right) \left( \int \left\vert \hat{F%
}(\rho ^{2},\mathbf{\xi }_{1})\right\vert ^{2}\rho d\rho \right) ^{\frac{1}{2%
}}\left( \int \left\vert \hat{G}(\rho ^{2})\right\vert ^{2}\rho d\rho
\right) ^{\frac{1}{2}} \\
&&\text{(H\"{o}lder in }\rho \text{)} \\
&\leqslant &C\left\Vert f(\cdot ,\mathbf{\xi }_{1})\right\Vert
_{L_{t}^{2}}\left\Vert g\right\Vert _{L_{t}^{2}}\left\{ \int \frac{%
\left\vert \mathbf{\xi }_{1}\right\vert ^{2}}{\left\vert \mathbf{\xi }%
_{2}^{\prime }\right\vert ^{2}}\sup_{\rho }\left( \int \frac{\rho ^{2}d%
\mathbf{\sigma }}{\rho \left\vert \mathbf{\xi }_{2}-\frac{\mathbf{\xi }_{1}-%
\mathbf{\xi }_{2}^{\prime }}{2}\right\vert ^{2}\left\vert \mathbf{\xi }_{2}+%
\frac{\mathbf{\xi }_{1}-\mathbf{\xi }_{2}^{\prime }}{2}\right\vert ^{2}}%
\right) d\mathbf{\xi }_{2}^{\prime }\right\}
\end{eqnarray*}%
However,%
\begin{eqnarray*}
&&\int \frac{\left\vert \mathbf{\xi }_{1}\right\vert ^{2}}{\left\vert 
\mathbf{\xi }_{2}^{\prime }\right\vert ^{2}}\sup_{\rho }\left( \int \frac{%
\rho ^{2}d\mathbf{\sigma }}{\rho \left\vert \mathbf{\xi }_{2}-\frac{\mathbf{%
\xi }_{1}-\mathbf{\xi }_{2}^{\prime }}{2}\right\vert ^{2}\left\vert \mathbf{%
\xi }_{2}+\frac{\mathbf{\xi }_{1}-\mathbf{\xi }_{2}^{\prime }}{2}\right\vert
^{2}}\right) d\mathbf{\xi }_{2}^{\prime } \\
&=&\int \frac{\left\vert \mathbf{\xi }_{1}\right\vert ^{2}}{\left\vert 
\mathbf{\xi }_{2}^{\prime }\right\vert ^{2}}\sup_{\rho }\left( \int \frac{%
\left\vert \mathbf{\xi }_{2}-\frac{\mathbf{\xi }_{1}-\mathbf{\xi }%
_{2}^{\prime }}{2}\right\vert ^{2}d\mathbf{\sigma }}{\left\vert \mathbf{\xi }%
_{2}-\frac{\mathbf{\xi }_{1}-\mathbf{\xi }_{2}^{\prime }}{2}\right\vert
\left\vert \mathbf{\xi }_{1}-\mathbf{\xi }_{2}-\mathbf{\xi }_{2}^{\prime
}\right\vert ^{2}\left\vert \mathbf{\xi }_{2}\right\vert ^{2}}\right) d%
\mathbf{\xi }_{2}^{\prime } \\
&&\text{(Reverse the change of variable in formula \ref{formula:change of
variable 3*3d sphere}.)} \\
&=&\left\vert \mathbf{\xi }_{1}\right\vert ^{2}\int \frac{d\mathbf{\xi }%
_{2}^{\prime }}{\left\vert \mathbf{\xi }_{2}^{\prime }\right\vert
^{2+2\varepsilon }}\sup_{\rho }\left( \int \frac{\left\vert \mathbf{\xi }%
_{2}-\frac{\mathbf{\xi }_{1}-\mathbf{\xi }_{2}^{\prime }}{2}\right\vert ^{2}d%
\mathbf{\sigma }}{\left\vert \mathbf{\xi }_{2}-\frac{\mathbf{\xi }_{1}-%
\mathbf{\xi }_{2}^{\prime }}{2}\right\vert \left\vert \mathbf{\xi }_{1}-%
\mathbf{\xi }_{2}-\mathbf{\xi }_{2}^{\prime }\right\vert ^{2-\varepsilon
}\left\vert \mathbf{\xi }_{2}\right\vert ^{2-\varepsilon }}\right) \\
&\leqslant &C\left\vert \mathbf{\xi }_{1}\right\vert ^{2}\int \frac{d\mathbf{%
\xi }_{2}^{\prime }}{\left\vert \mathbf{\xi }_{2}^{\prime }\right\vert
^{2+2\varepsilon }\left\vert \mathbf{\xi }_{1}-\mathbf{\xi }_{2}^{\prime
}\right\vert ^{3-2\varepsilon }}\text{ (Second part of Lemma \ref%
{Lemma:MateiLemmaForIntegrals})} \\
&\leqslant &C.
\end{eqnarray*}

In the above calculation, the $\mathbf{\sigma }$ in the first line lives on
the unit sphere centered at the origin while the $\mathbf{\sigma }$ in the
second line is on a unit sphere centered at $\frac{\mathbf{\xi }_{1}-\mathbf{%
\xi }_{2}^{\prime }}{2}$. We use the same symbol because Lebesgue measure is
translation invariant.

Thus,%
\begin{equation*}
\left\vert \int J(f)(t,\mathbf{\xi }_{1})\overline{g(t)}dt\right\vert
\leqslant C\left\Vert f(\cdot ,\mathbf{\xi }_{1})\right\Vert
_{L_{t}^{2}}\left\Vert g\right\Vert _{L_{t}^{2}}.
\end{equation*}
\end{proof}

\begin{remark}
Because the integral $I(h)$ is also symmetric in $\mathbf{\xi }_{2}$ and $%
\mathbf{\xi }_{2}^{\prime }$ when the "$\pm $" in equation \ref{eqn:3*3d} is
"$+$"$,$ we have acquired the estimate in that case. In Case II, we will
assume that "$\pm $" is "$-$".
\end{remark}

\subsection{Case II: $I(h)$ restricted to the region $\left\vert \mathbf{%
\protect\xi }_{2}^{\prime }\right\vert >\left\vert \mathbf{\protect\xi }%
_{2}\right\vert $ with integration order $d\mathbf{\protect\xi }_{2}^{\prime
}$ prior to $d\mathbf{\protect\xi }_{2}$}

This time we write the phase function to be%
\begin{eqnarray*}
&&\left( \mathbf{\xi }_{1}-\mathbf{\xi }_{2}-\mathbf{\xi }_{2}^{\prime
}\right) ^{T}A_{t}\left( \mathbf{\xi }_{1}-\mathbf{\xi }_{2}-\mathbf{\xi }%
_{2}^{\prime }\right) +\mathbf{\xi }_{2}^{T}A_{t}\mathbf{\xi }_{2}-\left( 
\mathbf{\xi }_{2}^{\prime }\right) ^{T}A_{t}\mathbf{\xi }_{2}^{\prime } \\
&=&\left( \mathbf{\xi }_{1}-\mathbf{\xi }_{2}\right) ^{T}A_{t}\left( \mathbf{%
\xi }_{1}-\mathbf{\xi }_{2}\right) -2\left( \mathbf{\xi }_{1}-\mathbf{\xi }%
_{2}\right) ^{T}A_{t}\mathbf{\xi }_{2}^{\prime }+\mathbf{\xi }_{2}^{T}A_{t}%
\mathbf{\xi }_{2} \\
&=&\phi (t,\mathbf{\xi }_{1},\mathbf{\xi }_{2})-2\left( \mathbf{\xi }_{1}-%
\mathbf{\xi }_{2}\right) ^{T}A_{t}\mathbf{\xi }_{2}^{\prime }.
\end{eqnarray*}%
and let%
\begin{equation*}
J\left( \overline{\hat{h}}\right) (t,\mathbf{\xi }_{1})=\int \frac{%
\left\vert \mathbf{\xi }_{1}\right\vert ^{2}e^{-2i\left( \mathbf{\xi }_{1}-%
\mathbf{\xi }_{2}\right) ^{T}A_{t}\mathbf{\xi }_{2}^{\prime }}e^{2i\left( 
\mathbf{\xi }_{1}-\mathbf{\xi }_{2}\right) ^{T}A_{t^{\prime }}\mathbf{\xi }%
_{2}^{\prime }}}{\left\vert \mathbf{\xi }_{1}-\mathbf{\xi }_{2}-\mathbf{\xi }%
_{2}^{\prime }\right\vert ^{2}\left\vert \mathbf{\xi }_{2}\right\vert
^{2}\left\vert \mathbf{\xi }_{2}^{\prime }\right\vert ^{2}}e^{-i\phi
(t^{\prime },\mathbf{\xi }_{1},\mathbf{\xi }_{2})}\overline{e^{-i\phi (t,%
\mathbf{\xi }_{1},\mathbf{\xi }_{2}^{\prime })}}\overline{\hat{h}(t^{\prime
},\mathbf{\xi }_{1})}dt^{\prime }d\mathbf{\xi }_{2}^{\prime }d\mathbf{\xi }%
_{2}.
\end{equation*}%
Again, we want to prove

\begin{proposition}
\begin{equation*}
\int \left\vert J(f)(t,\mathbf{\xi }_{1})\right\vert ^{2}dt\leqslant
C\left\Vert f(\cdot ,\mathbf{\xi }_{1})\right\Vert _{L_{t}^{2}}^{2}
\end{equation*}%
where $C$ is independent of $f$ or $\mathbf{\xi }_{1}.$
\end{proposition}

\begin{proof}
We calculate%
\begin{eqnarray*}
&&\left\vert \int J(f)(t,\mathbf{\xi }_{1})\overline{g(t)}dt\right\vert \\
&=&\left\vert \int \frac{\left\vert \mathbf{\xi }_{1}\right\vert
^{2}e^{-2i\left( \mathbf{\xi }_{1}-\mathbf{\xi }_{2}\right) ^{T}A_{t}\mathbf{%
\xi }_{2}^{\prime }}e^{2i\left( \mathbf{\xi }_{1}-\mathbf{\xi }_{2}\right)
^{T}A_{t^{\prime }}\mathbf{\xi }_{2}^{\prime }}}{\left\vert \mathbf{\xi }%
_{1}-\mathbf{\xi }_{2}-\mathbf{\xi }_{2}^{\prime }\right\vert ^{2}\left\vert 
\mathbf{\xi }_{2}\right\vert ^{2}\left\vert \mathbf{\xi }_{2}^{\prime
}\right\vert ^{2}}\left( e^{-i\phi (t^{\prime },\mathbf{\xi }_{1},\mathbf{%
\xi }_{2})}f(t^{\prime },\mathbf{\xi }_{1})\right) \left( \overline{%
e^{-i\phi (t,\mathbf{\xi }_{1},\mathbf{\xi }_{2})}g(t)}\right) dtdt^{\prime
}d\mathbf{\xi }_{2}^{\prime }d\mathbf{\xi }_{2}\right\vert \\
&=&\left\vert \int \frac{\left( \int e^{-2i\left( \mathbf{\xi }_{1}-\mathbf{%
\xi }_{2}\right) ^{T}A_{t}\mathbf{\xi }_{2}^{\prime }}\left( \overline{%
e^{-i\phi (t,\mathbf{\xi }_{1},\mathbf{\xi }_{2})}g(t)}\right) dt\right)
\left( \int e^{2i\left( \mathbf{\xi }_{1}-\mathbf{\xi }_{2}\right)
^{T}A_{t^{\prime }}\mathbf{\xi }_{2}}\left( e^{-i\phi (t^{\prime },\mathbf{%
\xi }_{1},\mathbf{\xi }_{2})}f(t^{\prime },\mathbf{\xi }_{1})\right)
dt^{\prime }\right) \left\vert \mathbf{\xi }_{1}\right\vert ^{2}d\mathbf{\xi 
}_{2}d\mathbf{\xi }_{2}^{\prime }}{\left\vert \mathbf{\xi }_{1}-\mathbf{\xi }%
_{2}-\mathbf{\xi }_{2}^{\prime }\right\vert ^{2}\left\vert \mathbf{\xi }%
_{2}\right\vert ^{2}\left\vert \mathbf{\xi }_{2}^{\prime }\right\vert ^{2}}%
\right\vert \\
&\leqslant &\int \frac{\left\vert \mathbf{\xi }_{1}\right\vert ^{2}d\mathbf{%
\xi }_{2}}{\left\vert \mathbf{\xi }_{2}\right\vert ^{2}}\int \frac{d\mathbf{%
\xi }_{2}^{\prime }}{\left\vert \mathbf{\xi }_{1}-\mathbf{\xi }_{2}-\mathbf{%
\xi }_{2}^{\prime }\right\vert ^{2}\left\vert \mathbf{\xi }_{2}^{\prime
}\right\vert ^{2}} \\
&&\left\vert \int e^{-2i\left( \mathbf{\xi }_{1}-\mathbf{\xi }_{2}\right)
^{T}A_{t}\mathbf{\xi }_{2}^{\prime }}\left( \overline{e^{-i\phi (t,\mathbf{%
\xi }_{1},\mathbf{\xi }_{2})}g(t)}\right) dt\right\vert \left\vert \int
e^{2i\left( \mathbf{\xi }_{1}-\mathbf{\xi }_{2}\right) ^{T}A_{t^{\prime }}%
\mathbf{\xi }_{2}^{\prime }}\left( e^{-i\phi (t^{\prime },\mathbf{\xi }_{1},%
\mathbf{\xi }_{2})}f(t^{\prime },\mathbf{\xi }_{1})\right) dt^{\prime
}\right\vert
\end{eqnarray*}%
Fix $\mathbf{\xi }_{1}-\mathbf{\xi }_{2}$ and $\mathbf{\xi }_{2}^{\prime }$,
write 
\begin{equation*}
\int e^{-2i\left( \mathbf{\xi }_{1}-\mathbf{\xi }_{2}\right) ^{T}A_{t}%
\mathbf{\xi }_{2}^{\prime }}\left( \overline{e^{-i\phi (t,\mathbf{\xi }_{1},%
\mathbf{\xi }_{2}^{\prime })}g(t)}\right) dt=\int e^{-2i\left\vert \mathbf{%
\xi }_{1}-\mathbf{\xi }_{2}\right\vert \omega ^{T}A_{t}\mathbf{\xi }%
_{2}^{\prime }}\left( \overline{e^{-i\phi (t,\mathbf{\xi }_{1},\mathbf{\xi }%
_{2}^{\prime })}g(t)}\right) dt
\end{equation*}%
where $\mathbf{\omega }=(\omega _{1},\omega _{2},\omega _{3})$ is a unit
vector in $\mathbb{R}^{3}$. Without loss of generality, we assume%
\begin{equation*}
\max \left\{ \left\vert \omega _{1}\right\vert ,\left\vert \omega
_{2}\right\vert ,\left\vert \omega _{3}\right\vert \right\} =\left\vert
\omega _{1}\right\vert
\end{equation*}%
which implies 
\begin{equation*}
\frac{1}{\sqrt{3}}\leqslant \left\vert \omega _{1}\right\vert \leqslant 1.
\end{equation*}%
Let us further assume that $\omega _{1}>0$ (the proof works exactly the same
for the $\omega _{1}<0$ case), then we can write 
\begin{eqnarray*}
\mathbf{\xi }_{2}^{\prime } &=&(x,0,0)+(0,y_{1},y_{2}) \\
u(t) &=&2\omega _{1}\int_{0}^{t}a_{1}(s)ds.
\end{eqnarray*}%
Again $u$ is invertible with 
\begin{equation*}
\frac{du}{dt}\geqslant \frac{2c_{0}}{\sqrt{3}}>0.
\end{equation*}%
So we have%
\begin{eqnarray*}
&&\int e^{-2i\left( \mathbf{\xi }_{1}-\mathbf{\xi }_{2}\right) ^{T}A_{t}%
\mathbf{\xi }_{2}^{\prime }}\left( \overline{e^{-i\phi (t,\mathbf{\xi }_{1},%
\mathbf{\xi }_{2}^{\prime })}g(t)}\right) dt \\
&=&\int e^{-2i\left\vert \mathbf{\xi }_{1}-\mathbf{\xi }_{2}\right\vert 
\mathbf{\omega }^{T}A_{t^{\prime }}\mathbf{\xi }_{2}^{\prime }}\left( 
\overline{e^{-i\phi (t,\mathbf{\xi }_{1},\mathbf{\xi }_{2}^{\prime })}g(t)}%
\right) dt \\
&=&\int e^{-iu\left( \omega _{1}\left\vert \mathbf{\xi }_{1}-\mathbf{\xi }%
_{2}\right\vert x\right) }\left( e^{-2i\left\vert \mathbf{\xi }_{1}-\mathbf{%
\xi }_{2}\right\vert (0,\omega _{2},\omega _{3})^{T}A_{t(u)}(0,y_{1},y_{2})}%
\overline{e^{-i\phi (t(u),\mathbf{\xi }_{1},\mathbf{\xi }_{2}^{\prime
})}g(t(u))}\left\vert \frac{dt}{du}\right\vert \right) du \\
&=&\overline{\hat{G}(-\omega _{1}\left\vert \mathbf{\xi }_{1}-\mathbf{\xi }%
_{2}\right\vert x)}
\end{eqnarray*}%
where%
\begin{equation*}
G(u)=\overline{e^{-2i\left\vert \mathbf{\xi }_{1}-\mathbf{\xi }%
_{2}\right\vert (0,\omega _{2},\omega _{3})^{T}A_{t(u)}(0,y_{1},y_{2})}}%
e^{-i\phi (t(u),\mathbf{\xi }_{1},\mathbf{\xi }_{2}^{\prime
})}g(t(u))\left\vert \frac{dt}{du}\right\vert
\end{equation*}%
which still has the property that 
\begin{equation*}
\int \left\vert G(u)\right\vert ^{2}du\leqslant \frac{\sqrt{3}}{2c_{0}}\int
\left\vert g(t)\right\vert ^{2}dt.
\end{equation*}%
Just as in case 1, this procedure hands us%
\begin{eqnarray*}
&&\left\vert \int J(f)(t,\mathbf{\xi }_{1})\overline{g(t)}dt\right\vert \\
&\leqslant &\int \frac{\left\vert \mathbf{\xi }_{1}\right\vert ^{2}d\mathbf{%
\xi }_{2}}{\left\vert \mathbf{\xi }_{2}\right\vert ^{2}}\int \frac{d\mathbf{%
\xi }_{2}^{\prime }}{\left\vert \mathbf{\xi }_{1}-\mathbf{\xi }_{2}-\mathbf{%
\xi }_{2}^{\prime }\right\vert ^{2}\left\vert \mathbf{\xi }_{2}^{\prime
}\right\vert ^{2}} \\
&&\left\vert \int e^{-2i\left( \mathbf{\xi }_{1}-\mathbf{\xi }_{2}\right)
^{T}A_{t}\mathbf{\xi }_{2}^{\prime }}\left( \overline{e^{-i\phi (t,\mathbf{%
\xi }_{1},\mathbf{\xi }_{2})}g(t)}\right) dt\right\vert \left\vert \int
e^{2i\left( \mathbf{\xi }_{1}-\mathbf{\xi }_{2}\right) ^{T}A_{t^{\prime }}%
\mathbf{\xi }_{2}^{\prime }}\left( e^{-i\phi (t^{\prime },\mathbf{\xi }_{1},%
\mathbf{\xi }_{2})}f(t^{\prime },\mathbf{\xi }_{1})\right) dt^{\prime
}\right\vert \\
&=&\int \left( \int \frac{dxdy_{1}dy_{2}}{\left\vert \mathbf{\xi }_{1}-%
\mathbf{\xi }_{2}-\mathbf{\xi }_{2}^{\prime }\right\vert ^{2}\left\vert 
\mathbf{\xi }_{2}^{\prime }\right\vert ^{2}}\left\vert \overline{\hat{G}%
(-\omega _{1}\left\vert \mathbf{\xi }_{1}-\mathbf{\xi }_{2}\right\vert x)}%
\hat{F}(-\omega _{1}\left\vert \mathbf{\xi }_{1}-\mathbf{\xi }%
_{2}\right\vert x,\mathbf{\xi }_{1})\right\vert \right) \frac{\left\vert 
\mathbf{\xi }_{1}\right\vert ^{2}}{\left\vert \mathbf{\xi }_{2}\right\vert
^{2}}d\mathbf{\xi }_{2} \\
&=&\int \left( \int \frac{dxdy_{1}dy_{2}}{\left\vert \mathbf{\xi }_{1}-%
\mathbf{\xi }_{2}-\mathbf{\xi }_{2}^{\prime }\right\vert ^{2}\left\vert 
\mathbf{\xi }_{2}^{\prime }\right\vert ^{2}}\left\vert \overline{\hat{G}(x)}%
\hat{F}(x,\mathbf{\xi }_{1})\right\vert \right) \frac{\left\vert \mathbf{\xi 
}_{1}\right\vert ^{2}}{\left\vert \omega _{1}\right\vert \left\vert \mathbf{%
\xi }_{1}-\mathbf{\xi }_{2}\right\vert \left\vert \mathbf{\xi }%
_{2}\right\vert ^{2}}d\mathbf{\xi }_{2} \\
&\leqslant &C\int \frac{\left\vert \mathbf{\xi }_{1}\right\vert ^{2}}{%
\left\vert \mathbf{\xi }_{1}-\mathbf{\xi }_{2}\right\vert \left\vert \mathbf{%
\xi }_{2}\right\vert ^{2}}\left( \sup_{x}\int \frac{dy_{1}dy_{2}}{\left\vert 
\mathbf{\xi }_{1}-\mathbf{\xi }_{2}-\mathbf{\xi }_{2}^{\prime }\right\vert
^{2}\left\vert \mathbf{\xi }_{2}^{\prime }\right\vert ^{2}}\right) \\
&&\left( \int \left\vert \hat{F}(x,\mathbf{\xi }_{1})\right\vert
^{2}dx\right) ^{\frac{1}{2}}\left( \int \left\vert \hat{G}(x)\right\vert
^{2}dx\right) ^{\frac{1}{2}}d\mathbf{\xi }_{2}\text{ (H\"{o}lder in }x\text{)%
} \\
&\leqslant &C\left\Vert f(\cdot ,\mathbf{\xi }_{1})\right\Vert
_{L_{t}^{2}}\left\Vert g\right\Vert _{L_{t}^{2}}\int \frac{\left\vert 
\mathbf{\xi }_{1}\right\vert ^{2}}{2\left\vert \mathbf{\xi }_{1}-\mathbf{\xi 
}_{2}\right\vert \left\vert \mathbf{\xi }_{2}\right\vert ^{2}}\left(
\sup_{x}\int \frac{dy_{1}dy_{2}}{\left\vert \mathbf{\xi }_{1}-\mathbf{\xi }%
_{2}-\mathbf{\xi }_{2}^{\prime }\right\vert ^{2}\left\vert \mathbf{\xi }%
_{2}^{\prime }\right\vert ^{2}}\right) d\mathbf{\xi }_{2}
\end{eqnarray*}%
The first part of Lemma \ref{Lemma:MateiLemmaForIntegrals} and the
restrictions that $\left\vert \mathbf{\xi }_{1}-\mathbf{\xi }_{2}-\mathbf{%
\xi }_{2}^{\prime }\right\vert >\left\vert \mathbf{\xi }_{2}\right\vert $
and $\left\vert \mathbf{\xi }_{2}^{\prime }\right\vert <\left\vert \mathbf{%
\xi }_{2}\right\vert $ show 
\begin{eqnarray*}
&&\int \frac{\left\vert \mathbf{\xi }_{1}\right\vert ^{2}}{2\left\vert 
\mathbf{\xi }_{1}-\mathbf{\xi }_{2}\right\vert \left\vert \mathbf{\xi }%
_{2}\right\vert ^{2}}\left( \sup_{x}\int \frac{dy_{1}dy_{2}}{\left\vert 
\mathbf{\xi }_{1}-\mathbf{\xi }_{2}-\mathbf{\xi }_{2}^{\prime }\right\vert
^{2}\left\vert \mathbf{\xi }_{2}^{\prime }\right\vert ^{2}}\right) d\mathbf{%
\xi }_{2} \\
&\leqslant &\int \frac{\left\vert \mathbf{\xi }_{1}\right\vert ^{2}}{%
2\left\vert \mathbf{\xi }_{1}-\mathbf{\xi }_{2}\right\vert \left\vert 
\mathbf{\xi }_{2}\right\vert ^{2+2\varepsilon }}\left( \sup_{x}\int \frac{%
dy_{1}dy_{2}}{\left\vert \mathbf{\xi }_{1}-\mathbf{\xi }_{2}-\mathbf{\xi }%
_{2}^{\prime }\right\vert ^{2-\varepsilon }\left\vert \mathbf{\xi }%
_{2}^{\prime }\right\vert ^{2-\varepsilon }}\right) d\mathbf{\xi }_{2} \\
&\leqslant &C\int \frac{\left\vert \mathbf{\xi }_{1}\right\vert ^{2}d\mathbf{%
\xi }_{2}}{2\left\vert \mathbf{\xi }_{1}-\mathbf{\xi }_{2}\right\vert
^{3-2\varepsilon }\left\vert \mathbf{\xi }_{2}\right\vert ^{2+2\varepsilon }}
\\
&\leqslant &C,
\end{eqnarray*}%
which finishes the proposition.
\end{proof}

\section{Proof of Theorem \protect\ref{Theorem:3*nd} when $n=2$ / 3*2d
Collapsing Estimate \label{Sec:ProofOf3*2d}}

By the proof of the $n=3$ case in Section \ref{Sec:ProofOf3*3d}, we only
need to show these two estimates:

\begin{itemize}
\item[Case I] Under the restrictions $\left\vert \mathbf{\xi }_{1}-\mathbf{%
\xi }_{2,old}-\mathbf{\xi }_{2}^{\prime }\right\vert >\left\vert \mathbf{\xi 
}_{2,old}\right\vert $ and $\left\vert \mathbf{\xi }_{2}^{\prime
}\right\vert <\left\vert \mathbf{\xi }_{2,old}\right\vert $, we have%
\begin{equation*}
\int \frac{\left\vert \mathbf{\xi }_{1}\right\vert }{\left\vert \mathbf{\xi }%
_{2}^{\prime }\right\vert }\sup_{\rho }\left( \int \frac{d\mathbf{\sigma (%
\mathbf{\xi }}_{2,new}\mathbf{)}}{\left\vert \mathbf{\xi }_{2,new}-\frac{%
\mathbf{\xi }_{1}-\mathbf{\xi }_{2}^{\prime }}{2}\right\vert \left\vert 
\mathbf{\xi }_{2,new}+\frac{\mathbf{\xi }_{1}-\mathbf{\xi }_{2}^{\prime }}{2}%
\right\vert }\right) d\mathbf{\xi }_{2}^{\prime }\leqslant C
\end{equation*}%
where $\mathbf{\mathbf{\xi }}_{2,new}$ and $\mathbf{\xi }_{2,old}$ are
related by formula \ref{formula:change of variable 3*3d sphere} and we write 
\begin{equation*}
\mathbf{\mathbf{\xi }}_{2,new}\mathbf{=}\rho \mathbf{\sigma }\text{ with }%
\mathbf{\sigma }\in \mathbb{S}^{1}.
\end{equation*}

\item[Case II] Under the restrictions $\left\vert \mathbf{\xi }_{1}-\mathbf{%
\xi }_{2}-\mathbf{\xi }_{2}^{\prime }\right\vert >\left\vert \mathbf{\xi }%
_{2}\right\vert $ and $\left\vert \mathbf{\xi }_{2}^{\prime }\right\vert
>\left\vert \mathbf{\xi }_{2}\right\vert $, we have%
\begin{equation*}
\int \frac{\left\vert \mathbf{\xi }_{1}\right\vert }{\left\vert \mathbf{\xi }%
_{1}-\mathbf{\xi }_{2}\right\vert \left\vert \mathbf{\xi }_{2}\right\vert }%
\left( \sup_{x}\int \frac{dy}{\left\vert \mathbf{\xi }_{1}-\mathbf{\xi }_{2}-%
\mathbf{\xi }_{2}^{\prime }\right\vert \left\vert \mathbf{\xi }_{2}^{\prime
}\right\vert }\right) d\mathbf{\xi }_{2}\leqslant C.
\end{equation*}%
where $\mathbf{\xi }_{2}^{\prime }=(x,y).$
\end{itemize}

Lemma \ref{Lemma:MateiLemmaForIntegrals} plays an important role in giving
the corresponding estimates in Section \ref{Sec:ProofOf3*3d}. In the 2d
case, the subsequent lemma provides its replacement.

\begin{lemma}
\label{Lemma:KeyLemmaFor3*2d}Let $\mathbf{\xi }\in \mathbb{R}^{2}$ and $L$
be a 1d line or circle in $\mathbb{R}^{2}$ with the usual induced line
element $dS$.

(1) Say $0<a,$ $b<1,$ $a+b>1,$ then there exists a $C$ independent of $L$
s.t.%
\begin{equation*}
\int_{L}\frac{dS(\mathbf{\eta })}{\left\vert \mathbf{\xi }-\mathbf{\eta }%
\right\vert ^{a}\left\vert \mathbf{\eta }\right\vert ^{b}}\leqslant \frac{C}{%
\left\vert \mathbf{\xi }\right\vert ^{a+b-1}}.
\end{equation*}

(2) Let $\varepsilon =\frac{1}{80}$, then 
\begin{equation*}
\sup_{\left\vert \mathbf{\eta }\right\vert }\left( \int_{\mathbb{S}^{1}}%
\frac{d\mathbf{\sigma (\eta )}}{\left\vert \mathbf{\xi }-\mathbf{\eta }%
\right\vert ^{1-\varepsilon }\left\vert \mathbf{\xi }+\mathbf{\eta }%
\right\vert ^{1-\varepsilon }}\right) \leqslant \frac{C}{\left\vert \mathbf{%
\xi }\right\vert ^{2-2\varepsilon }}.
\end{equation*}
\end{lemma}

\begin{proof}
We will show the second part in the end of this section. The first part
shares the exact same proof with Lemma 2.2 in \cite{KlainermanAndMachedon}.
\end{proof}

\subsection{Proof of Case I}

The change of variable \ref{formula:change of variable 3*3d sphere} turns
the restrictions into 
\begin{eqnarray*}
\left\vert \mathbf{\xi }_{2,new}-\frac{\mathbf{\xi }_{1}-\mathbf{\xi }%
_{2}^{\prime }}{2}\right\vert &=&\left\vert \mathbf{\xi }_{1}-\mathbf{\xi }%
_{2,old}-\mathbf{\xi }_{2}^{\prime }\right\vert >\left\vert \mathbf{\xi }%
_{2,old}\right\vert >\left\vert \mathbf{\xi }_{2}^{\prime }\right\vert , \\
\left\vert \mathbf{\xi }_{2,new}+\frac{\mathbf{\xi }_{1}-\mathbf{\xi }%
_{2}^{\prime }}{2}\right\vert &=&\left\vert \mathbf{\xi }_{2,old}\right\vert
>\left\vert \mathbf{\xi }_{2}^{\prime }\right\vert .
\end{eqnarray*}%
Notice that $\mathbf{\mathbf{\xi }}_{2,new}\mathbf{=}\rho \mathbf{\sigma ,}$
we in fact have%
\begin{eqnarray*}
&&\int \frac{\left\vert \mathbf{\xi }_{1}\right\vert }{\left\vert \mathbf{%
\xi }_{2}^{\prime }\right\vert }\sup_{\rho }\left( \int \frac{d\mathbf{%
\sigma }\left( \mathbf{\xi }_{2,new}\right) }{\left\vert \mathbf{\xi }%
_{2,new}-\frac{\mathbf{\xi }_{1}-\mathbf{\xi }_{2}^{\prime }}{2}\right\vert
\left\vert \mathbf{\xi }_{2,new}+\frac{\mathbf{\xi }_{1}-\mathbf{\xi }%
_{2}^{\prime }}{2}\right\vert }\right) d\mathbf{\xi }_{2}^{\prime } \\
&\leqslant &\int \frac{\left\vert \mathbf{\xi }_{1}\right\vert }{\left\vert 
\mathbf{\xi }_{2}^{\prime }\right\vert ^{1+2\varepsilon }}\sup_{\rho }\left(
\int_{\mathbb{S}^{1}}\frac{d\mathbf{\sigma }\left( \mathbf{\xi }%
_{2,new}\right) }{\left\vert \mathbf{\xi }_{2,new}-\frac{\mathbf{\xi }_{1}-%
\mathbf{\xi }_{2}^{\prime }}{2}\right\vert ^{1-\varepsilon }\left\vert 
\mathbf{\xi }_{2,new}+\frac{\mathbf{\xi }_{1}-\mathbf{\xi }_{2}^{\prime }}{2}%
\right\vert ^{1-\varepsilon }}\right) d\mathbf{\xi }_{2}^{\prime } \\
&\leqslant &\text{ }C\left\vert \mathbf{\xi }_{1}\right\vert \int \frac{1}{%
\left\vert \mathbf{\xi }_{2}^{\prime }\right\vert ^{1+2\varepsilon }}\frac{1%
}{\left\vert \mathbf{\xi }_{1}-\mathbf{\xi }_{2}^{\prime }\right\vert
^{2-2\varepsilon }}d\mathbf{\xi }_{2}^{\prime }\text{ (Second part of Lemma %
\ref{Lemma:KeyLemmaFor3*2d})} \\
&\leqslant &C.
\end{eqnarray*}

\subsection{Proof of Case II}

Recall that $\mathbf{\xi }_{2}^{\prime }=(x,y)$, we estimate

\begin{eqnarray*}
&&\int \frac{\left\vert \mathbf{\xi }_{1}\right\vert }{\left\vert \mathbf{%
\xi }_{1}-\mathbf{\xi }_{2}\right\vert \left\vert \mathbf{\xi }%
_{2}\right\vert }\left( \sup_{x}\int \frac{dy}{\left\vert \mathbf{\xi }_{1}-%
\mathbf{\xi }_{2}-\mathbf{\xi }_{2}^{\prime }\right\vert \left\vert \mathbf{%
\xi }_{2}^{\prime }\right\vert }\right) d\mathbf{\xi }_{2} \\
&\leqslant &\int \frac{\left\vert \mathbf{\xi }_{1}\right\vert }{\left\vert 
\mathbf{\xi }_{1}-\mathbf{\xi }_{2}\right\vert \left\vert \mathbf{\xi }%
_{2}\right\vert ^{1+2\varepsilon }}\left( \sup_{x}\int \frac{dy}{\left\vert 
\mathbf{\xi }_{1}-\mathbf{\xi }_{2}-\mathbf{\xi }_{2}^{\prime }\right\vert
^{1-\varepsilon }\left\vert \mathbf{\xi }_{2}^{\prime }\right\vert
^{1-\varepsilon }}\right) d\mathbf{\xi }_{2} \\
&\leqslant &C\left\vert \mathbf{\xi }_{1}\right\vert \int \frac{1}{%
\left\vert \mathbf{\xi }_{1}-\mathbf{\xi }_{2}\right\vert ^{2-2\varepsilon
}\left\vert \mathbf{\xi }_{2}\right\vert ^{1+2\varepsilon }}d\mathbf{\xi }%
_{2}\text{ (First part of Lemma \ref{Lemma:KeyLemmaFor3*2d})} \\
&\leqslant &C.
\end{eqnarray*}

\subsection{Proof of the Second Part of Lemma \protect\ref%
{Lemma:KeyLemmaFor3*2d}}

Due to%
\begin{equation*}
\left\vert \mathbf{\xi }\right\vert \leqslant \left\vert \mathbf{\xi }-%
\mathbf{\eta }\right\vert +\left\vert \mathbf{\xi }+\mathbf{\eta }%
\right\vert ,
\end{equation*}%
we can separate the integral as

\begin{eqnarray*}
&&\sup_{\left\vert \mathbf{\eta }\right\vert }\left( \int_{\mathbb{S}^{1}}%
\frac{d\mathbf{\sigma (\eta )}}{\left\vert \mathbf{\xi }-\mathbf{\eta }%
\right\vert ^{1-\varepsilon }\left\vert \mathbf{\xi }+\mathbf{\eta }%
\right\vert ^{1-\varepsilon }}\right) \\
&\leqslant &\sup_{\left\vert \mathbf{\eta }\right\vert }\left( \int_{\mathbb{%
S}^{1}\text{ and }\left\vert \mathbf{\xi }-\mathbf{\eta }\right\vert
\geqslant \frac{\left\vert \mathbf{\xi }\right\vert }{2}}\right)
+\sup_{\left\vert \mathbf{\eta }\right\vert }\left( \int_{\mathbb{S}^{1}%
\text{ and }\left\vert \mathbf{\xi }+\mathbf{\eta }\right\vert \geqslant 
\frac{\left\vert \mathbf{\xi }\right\vert }{2}}\right) .
\end{eqnarray*}%
We will only show%
\begin{equation*}
\sup_{\left\vert \mathbf{\eta }\right\vert }\left( \int_{\mathbb{S}^{1}\text{
and }\left\vert \mathbf{\xi }+\mathbf{\eta }\right\vert \geqslant \frac{%
\left\vert \mathbf{\xi }\right\vert }{2}}\frac{d\mathbf{\sigma (\eta )}}{%
\left\vert \mathbf{\xi }-\mathbf{\eta }\right\vert ^{1-\varepsilon
}\left\vert \mathbf{\xi }+\mathbf{\eta }\right\vert ^{1-\varepsilon }}%
\right) \leqslant \frac{C}{\left\vert \mathbf{\xi }\right\vert
^{2-2\varepsilon }}
\end{equation*}%
since the other part is similar. It is clear that

\begin{equation}
\sup_{\left\vert \mathbf{\eta }\right\vert }\left( \int_{\mathbb{S}^{1}\text{
and }\left\vert \mathbf{\xi }+\mathbf{\eta }\right\vert \geqslant \frac{%
\left\vert \mathbf{\xi }\right\vert }{2}}\frac{d\mathbf{\sigma (\eta )}}{%
\left\vert \mathbf{\xi }-\mathbf{\eta }\right\vert ^{1-\varepsilon
}\left\vert \mathbf{\xi }+\mathbf{\eta }\right\vert ^{1-\varepsilon }}%
\right) \leqslant \frac{C}{\left\vert \mathbf{\xi }\right\vert
^{1-\varepsilon }}\sup_{\left\vert \mathbf{\eta }\right\vert }\left( \int_{%
\mathbb{S}^{1}}\frac{d\mathbf{\sigma (\eta )}}{\left\vert \mathbf{\xi }-%
\mathbf{\eta }\right\vert ^{1-\varepsilon }}\right) .
\label{estimate:3*2d middle 1}
\end{equation}%
Rotate $\mathbb{S}^{1}$ such that $\mathbf{\xi }$ is on the positive $x$
axis, then write $\mathbf{\eta }=\rho e^{i\theta }$ for $(\rho \cos \theta
,\rho \sin \theta )$ and observe:

\begin{itemize}
\item When $\theta \in \lbrack 0,\frac{\pi }{2}]\cup \lbrack \frac{3\pi }{2}%
,2\pi ],$%
\begin{equation*}
\left\vert \rho e^{i\theta }-(\left\vert \mathbf{\xi }\right\vert
,0)\right\vert \geqslant \left\vert \mathbf{\xi }\right\vert \left\vert \sin
\theta \right\vert
\end{equation*}%
because $\left\vert \mathbf{\xi }\right\vert \left\vert \sin \theta
\right\vert $ is the distance between the point $(\left\vert \mathbf{\xi }%
\right\vert ,0)$ and the line ($angle=\theta $).

\item When $\theta \in \lbrack \frac{\pi }{2},\frac{3\pi }{2}],$%
\begin{equation*}
\left\vert \rho e^{i\theta }-(\left\vert \mathbf{\xi }\right\vert
,0)\right\vert \geqslant \left\vert \mathbf{\xi }\right\vert
\end{equation*}%
because $\rho e^{i\theta }-(\left\vert \mathbf{\xi }\right\vert ,0)$ is the
longest edge in the obtuse triangle which consists of $\rho e^{i\theta }%
\mathbf{,}$ $(\left\vert \mathbf{\xi }\right\vert ,0)$ and $\rho e^{i\theta
}-(\left\vert \mathbf{\xi }\right\vert ,0).$
\end{itemize}

Insert these two elementary observations into estimate \ref{estimate:3*2d
middle 1}, we have 
\begin{eqnarray*}
&&\sup_{\left\vert \mathbf{\eta }\right\vert }\left( \int_{\mathbb{S}^{1}%
\text{ and }\left\vert \mathbf{\xi }+\mathbf{\eta }\right\vert \geqslant 
\frac{\left\vert \mathbf{\xi }\right\vert }{2}}\frac{d\mathbf{\sigma (\eta )}%
}{\left\vert \mathbf{\xi }-\mathbf{\eta }\right\vert ^{1-\varepsilon
}\left\vert \mathbf{\xi }+\mathbf{\eta }\right\vert ^{1-\varepsilon }}\right)
\\
&\leqslant &\frac{C}{\left\vert \mathbf{\xi }\right\vert ^{1-\varepsilon }}%
\sup_{\left\vert \mathbf{\eta }\right\vert }\left( \int_{\mathbb{S}^{1}}%
\frac{d\mathbf{\sigma (\eta )}}{\left\vert \mathbf{\xi }-\mathbf{\eta }%
\right\vert ^{1-\varepsilon }}\right) \\
&\leqslant &\frac{C}{\left\vert \mathbf{\xi }\right\vert ^{1-\varepsilon }}%
\left[ \sup_{\rho }\left( \int_{\frac{\pi }{2}}^{\frac{3\pi }{2}}\frac{%
d\theta }{\left\vert \rho e^{i\theta }-(\left\vert \mathbf{\xi }\right\vert
,0)\right\vert ^{1-\varepsilon }}\right) +2\sup_{\rho }\left( \int_{0}^{%
\frac{\pi }{2}}\frac{d\theta }{\left\vert \rho e^{i\theta }-(\left\vert 
\mathbf{\xi }\right\vert ,0)\right\vert ^{1-\varepsilon }}\right) \right] \\
&\leqslant &\frac{C}{\left\vert \mathbf{\xi }\right\vert ^{1-\varepsilon }}%
\left[ \left( \int_{\frac{\pi }{2}}^{\frac{3\pi }{2}}\frac{d\theta }{%
\left\vert \mathbf{\xi }\right\vert ^{1-\varepsilon }}\right) +2\left(
\int_{0}^{\frac{\pi }{2}}\frac{d\theta }{\left\vert \left\vert \mathbf{\xi }%
\right\vert \sin \theta \right\vert ^{1-\varepsilon }}\right) \right] \\
&\leqslant &\frac{C}{\left\vert \mathbf{\xi }\right\vert ^{2-2\varepsilon }}.
\end{eqnarray*}

To show the other part, namely 
\begin{equation*}
\sup_{\left\vert \mathbf{\eta }\right\vert }\left( \int_{\mathbb{S}^{1}\text{
and }\left\vert \mathbf{\xi }-\mathbf{\eta }\right\vert \geqslant \frac{%
\left\vert \mathbf{\xi }\right\vert }{2}}\frac{d\mathbf{\sigma (\eta )}}{%
\left\vert \mathbf{\xi }-\mathbf{\eta }\right\vert ^{1-\varepsilon
}\left\vert \mathbf{\xi }+\mathbf{\eta }\right\vert ^{1-\varepsilon }}%
\right) \leqslant \frac{C}{\left\vert \mathbf{\xi }\right\vert
^{2-2\varepsilon }},
\end{equation*}%
one just needs to notice%
\begin{equation*}
\left\vert \mathbf{\xi }+\mathbf{\eta }\right\vert =\left\vert (\left\vert 
\mathbf{\xi }\right\vert ,0)-\rho e^{i\left( \theta +\pi \right)
}\right\vert ,
\end{equation*}%
then one can proceed as above. Therefore we conclude the proof of the second
part of Lemma \ref{Lemma:KeyLemmaFor3*2d}.

\section{The Lens Transform / Preparation for Theorem \protect\ref%
{Theorem:Collapsing for GP}\label{Sec:TheLensTransform}}

From now on, we enter the proof of Theorems \ref{Theorem:Collapsing for GP}
and \ref{Theorem:Uniqueness of GP}. We set $n=3$ until Section \ref%
{Sec:ProofOfBEC}. In this section, we set up the tools involved in the proof
of Theorem \ref{Theorem:Collapsing for GP}. We build the lens transform we
need and state the related properties. For simplicity of notations, we write 
$U^{\left( k+1\right) }(\tau ;s)$ to be the solution operator of equation %
\ref{eqn:homogeneous hierarchy with anisotropic traps} and $U_{\mathbf{y}%
}(\tau ;s)$ to be the solution operator of%
\begin{eqnarray*}
\left( i\partial _{\tau }-\frac{1}{2}H_{\mathbf{y}}(\tau )\right) u &=&0 \\
u(s,\mathbf{y}) &=&u_{s}(\mathbf{y}).
\end{eqnarray*}%
i.e. $U^{\left( k+1\right) }(\tau ;s)\gamma _{0}^{(k+1)}$ solves equation %
\ref{eqn:homogeneous hierarchy with anisotropic traps}. By definition, 
\begin{equation*}
U^{\left( k\right) }(\tau ;s)=\dprod\limits_{j=1}^{k}\left( U_{\mathbf{y}%
_{j}}(\tau ;s)U_{\mathbf{y}_{j}^{\prime }}(-\tau ;-s)\right) .
\end{equation*}

To be specific, we need this version of the generalized lens transform:

\begin{proposition}
\label{Proposition:TheLensTransformNeeded}There is an operator $L_{\mathbf{x}%
}(t)$ which satisfies the hypothesis in Theorem \ref{Theorem:3*nd} such that%
\begin{eqnarray*}
&&U^{(k+1)}(\tau ;0)\gamma _{0}^{(k+1)} \\
&=&\tprod_{j=1}^{k+1}\left( \dprod\limits_{l=1}^{3}\frac{e^{i\frac{\dot{\beta%
}_{l}(\tau )}{\beta _{l}(\tau )}\frac{\left( \left\vert y_{j,l}\right\vert
^{2}-\left\vert y_{j,l}^{\prime }\right\vert ^{2}\right) }{2}}}{\beta
_{l}(\tau )}\right) \\
&&u^{(k+1)}(\frac{\alpha _{1}(\tau )}{\beta _{1}(\tau )},\frac{y_{1,1}}{%
\beta _{1}(\tau )},\frac{y_{1,2}}{\beta _{2}(\tau )},\frac{y_{1,3}}{\beta
_{3}(\tau )},...,\frac{y_{k+1,1}}{\beta _{1}(\tau )},\frac{y_{k+1,2}}{\beta
_{2}(\tau )},\frac{y_{k+1,3}}{\beta _{3}(\tau )}; \\
&&\frac{y_{1,1}^{\prime }}{\beta _{1}(\tau )},\frac{y_{1,2}^{\prime }}{\beta
_{2}(\tau )},\frac{y_{1,3}^{\prime }}{\beta _{3}(\tau )},...,\frac{%
y_{k+1,1}^{\prime }}{\beta _{1}(\tau )},\frac{y_{k+1,2}^{\prime }}{\beta
_{2}(\tau )},\frac{y_{k+1,3}^{\prime }}{\beta _{3}(\tau )})
\end{eqnarray*}%
in $[-T_{0},T_{0}],$ where $\alpha _{l}$ and $\beta _{l}$ are defined as in
Claim \ref{Claim:Properties of Alpha and Beta}, and $u^{(k+1)}(t,%
\overrightarrow{\mathbf{x}_{k+1}};\overrightarrow{\mathbf{x}_{k+1}^{\prime }}%
)$ is the solution of%
\begin{eqnarray*}
\left( i\partial _{t}+L_{\overrightarrow{\mathbf{x}_{k+1}}}(t)-L_{%
\overrightarrow{\mathbf{x}_{k+1}^{\prime }}}(t)\right) u^{(k+1)} &=&0\text{
in }\mathbb{R}^{\left( 6k+6\right) +1} \\
u^{(k+1)}(0,\overrightarrow{\mathbf{x}_{k+1}};\overrightarrow{\mathbf{x}%
_{k+1}^{\prime }}) &=&\gamma _{0}^{(k+1)}.
\end{eqnarray*}
\end{proposition}

The proposition will be a corollary of a sequence of claims.

\begin{claim}
\label{Claim:Properties of Alpha and Beta}Assuming Conditions \ref%
{Condition:EvenExtension} and \ref{Condition:FastSwitch}, for $l=1,2,3$, the
system%
\begin{eqnarray}
\ddot{\alpha}_{l}(\tau )+\eta _{l}(\tau )\alpha _{l}(\tau ) &=&0,\alpha
_{l}(0)=0,\dot{\alpha}_{l}(0)=1,  \label{eqn: Alpha and Beta} \\
\ddot{\beta}_{l}(\tau )+\eta _{l}(\tau )\beta _{l}(\tau ) &=&0,\beta
_{l}(0)=1,\dot{\beta}_{l}(0)=0.  \notag
\end{eqnarray}%
defines an odd $\alpha _{l}$ and an even $\beta _{l}\in C^{2}(\mathbb{R})$
with the following properties

(1) $\beta _{l}$ is nonzero in $[-T_{0},T_{0}];$

(2) The Wronskian of $\alpha _{l}$ and $\beta _{l}$ is constant $1$ i.e.%
\begin{equation*}
\dot{\alpha}_{l}(\tau )\beta _{l}(\tau )-\alpha _{l}(\tau )\dot{\beta}%
_{l}(\tau )=1;
\end{equation*}

(3) The odd function%
\begin{equation*}
\upsilon _{l}(\tau )=\frac{\alpha _{l}(\tau )}{\beta _{l}(\tau )}
\end{equation*}%
is invertible in $[-T_{0},T_{0}]$ because 
\begin{equation*}
\dot{\upsilon}_{l}(\tau )=\frac{1}{\left( \beta _{l}(\tau )\right) ^{2}}>0%
\text{ in }[-T_{0},T_{0}].
\end{equation*}
\end{claim}

\begin{proof}
We show (1) only since all other statements are fairly trivial.

Suppose $\beta _{l}(\tau _{0})=0$ for some $\tau _{0}$ in $[-T_{0},T_{0}]$
then $\beta _{l}(-\tau _{0})=0$ via $\beta _{l}$ is even. Of course $\tau
_{0}\neq 0$ because $\beta _{l}(0)=1.$ Notice that $\cos \left( \tau \sqrt{%
\sup_{\tau }\left\vert \eta _{l}(\tau )\right\vert }\right) $ is a
nontrivial solution of%
\begin{equation*}
\ddot{v}(\tau )+\sup_{\tau }\left\vert \eta _{l}(\tau )\right\vert v(\tau
)=0.
\end{equation*}%
Since $\cos \left( \tau \sqrt{\sup_{\tau }\left\vert \eta _{l}(\tau
)\right\vert }\right) $ is not a multiple of $\beta _{l},$ $\cos \left( \tau 
\sqrt{\sup_{\tau }\left\vert \eta _{l}(\tau )\right\vert }\right) $ must
have at least one zero in $[-\tau _{0},\tau _{0}]$ due to the Sturm--Picone
comparison theorem. But this creates a contradiction.
\end{proof}

Though Claim \ref{Claim:Properties of Alpha and Beta} is elementary, its
consequences lying below make our procedure well-defined.

\begin{definition}
(A reminder of the norm) Let $\beta _{l}$ be defined via equation \ref{eqn:
Alpha and Beta}. We define 
\begin{equation*}
P_{\mathbf{y}}(\tau )=%
\begin{pmatrix}
i\beta _{1}(\tau )\frac{\partial }{\partial y_{1}}+\dot{\beta}_{1}(\tau
)y_{1} \\ 
i\beta _{2}(\tau )\frac{\partial }{\partial y_{2}}+\dot{\beta}_{2}(\tau
)y_{2} \\ 
i\beta _{3}(\tau )\frac{\partial }{\partial y_{3}}+\dot{\beta}_{3}(\tau
)y_{3}%
\end{pmatrix}%
\end{equation*}%
and%
\begin{equation*}
R_{\tau }^{k}=\tprod\nolimits_{j=1}^{k}P_{\mathbf{y}_{j}}(\tau )P_{\mathbf{y}%
_{j}^{\prime }}(-\tau ).
\end{equation*}
\end{definition}

\begin{lemma}
\label{Lemma:Monentum}$P_{\mathbf{y}}(\tau )$ commutes with the linear
operator%
\begin{equation*}
i\partial _{\tau }-\frac{1}{2}\left( -\triangle _{\mathbf{y}_{k}}+\eta (\tau
)\left\vert \mathbf{y}_{k}\right\vert ^{2}\right) .
\end{equation*}%
Moreover,%
\begin{equation*}
P_{\mathbf{y}}(\tau )U_{\mathbf{y}}(\tau ;s)f=U_{\mathbf{y}}(\tau ;s)P_{%
\mathbf{y}}(s)f.
\end{equation*}
\end{lemma}

\begin{lemma}
\label{Lemma:GLensTransform}Say $K_{1}(t,x_{0},y_{0})$ is the Green's
function of the 1d free Schr\"{o}dinger equation%
\begin{equation*}
\left( i\partial _{t}+\frac{1}{2}\frac{\partial ^{2}}{\partial x^{2}}\right)
v=0,
\end{equation*}%
then%
\begin{equation}
U_{\mathbf{y}}(\tau ;0)u_{0}=\left( \dprod\limits_{l=1}^{3}\frac{e^{i\frac{%
\dot{\beta}_{l}(\tau )}{\beta _{l}(\tau )}\frac{y_{l}^{2}}{2}}}{\left( \beta
_{l}(\tau )\right) ^{\frac{1}{2}}}\right) \int \left(
\dprod\limits_{l=1}^{3}K_{1}(\frac{\alpha _{l}(\tau )}{\beta _{l}(\tau )},%
\frac{y_{l}}{\beta _{l}(\tau )},y_{0l})\right)
u_{0}(y_{01},y_{02},y_{03})dy_{01}dy_{02}dy_{03},
\label{formula:3d generalized lens transform}
\end{equation}%
valid in the interval $[-T,T]$ in which $\eta _{l}$ are Lipschitzian and $%
\beta _{l}(\tau )\neq 0.$
\end{lemma}

\begin{proof}
Carles computed the isotropic case of formula \ref{formula:3d generalized
lens transform} in \cite{Carles}. We include a proof of Lemmas \ref%
{Lemma:Monentum}\ and \ref{Lemma:GLensTransform}\ using the metaplectic
representation in the appendix.
\end{proof}

We can now prove Proposition \ref{Proposition:TheLensTransformNeeded}. On
the one hand, via Claim \ref{Claim:Properties of Alpha and Beta}, we can
invert 
\begin{equation*}
t(\tau )=\upsilon _{1}(\tau )=\frac{\alpha _{1}(\tau )}{\beta _{1}(\tau )}%
\text{ in }[-T_{0},T_{0}].
\end{equation*}%
Therefore, the integral part of formula \ref{formula:3d generalized lens
transform} 
\begin{equation*}
\phi (t,\mathbf{x})=\int \left( K_{1}(t,x_{1},y_{01})K_{1}(\upsilon
_{2}(\upsilon _{1}^{-1}(t)),x_{2},y_{02})K_{1}(\upsilon _{3}(\upsilon
_{1}^{-1}(t)),x_{3},y_{03})\right)
u_{0}(y_{01},y_{02},y_{03})dy_{01}dy_{02}dy_{03}
\end{equation*}%
in fact solves%
\begin{eqnarray*}
\left( i\partial _{t}+\widetilde{L_{\mathbf{x}}}(t)\right) \phi &=&0\text{
in }\mathbb{R}^{3}\times \lbrack -\upsilon _{1}^{-1}(T_{0}),\upsilon
_{1}^{-1}(T_{0})] \\
\phi (0,\mathbf{x}) &=&u_{0},
\end{eqnarray*}%
where%
\begin{equation*}
\widetilde{L_{\mathbf{x}}}(t)=\frac{1}{2}\frac{\partial ^{2}}{\partial
x_{1}^{2}}+\frac{1}{2}\frac{\beta _{1}^{2}(\upsilon _{1}^{-1}(t))}{\beta
_{2}^{2}(\upsilon _{1}^{-1}(t))}\frac{\partial ^{2}}{\partial x_{2}^{2}}+%
\frac{1}{2}\frac{\beta _{1}^{2}(\upsilon _{1}^{-1}(t))}{\beta
_{3}^{2}(\upsilon _{1}^{-1}(t))}\frac{\partial ^{2}}{\partial x_{3}^{2}}.
\end{equation*}%
On the other hand, plugging $-\tau $ into formula \ref{formula:3d
generalized lens transform} yields 
\begin{equation*}
U_{\mathbf{y}}(-\tau ;0)u_{0}=\left( \dprod\limits_{l=1}^{3}\frac{e^{-i\frac{%
\dot{\beta}_{l}(\tau )}{\beta _{l}(\tau )}\frac{y_{l}^{2}}{2}}}{\left( \beta
_{l}(\tau )\right) ^{\frac{1}{2}}}\right) \int \left(
\dprod\limits_{l=1}^{3}K_{1}(-\frac{\alpha _{l}(\tau )}{\beta _{l}(\tau )},%
\frac{y_{l}}{\beta _{l}(\tau )},y_{0l})\right)
u_{0}(y_{01},y_{02},y_{03})dy_{01}dy_{02}dy_{03}
\end{equation*}%
because $\alpha _{l}$ and $\dot{\beta}_{l}$ are odd while $\beta _{l}$ are
even.

Whence in $[-T_{0},T_{0}]$%
\begin{eqnarray*}
U^{(k+1)}(\tau ;0)\gamma _{0}^{(k+1)} &=&\dprod\limits_{j=1}^{k+1}\left( U_{%
\mathbf{y}_{j}}(\tau ;0)U_{\mathbf{y}_{j}^{\prime }}(-\tau ;0)\right) \gamma
_{0}^{(k+1)} \\
&=&\tprod_{j=1}^{k+1}\left( \dprod\limits_{l=1}^{3}\frac{e^{i\frac{\dot{\beta%
}_{l}(\tau )}{\beta _{l}(\tau )}\frac{\left( \left\vert y_{j,l}\right\vert
^{2}-\left\vert y_{j,l}^{\prime }\right\vert ^{2}\right) }{2}}}{\beta
_{l}(\tau )}\right) \\
&&u^{(k+1)}(\frac{\alpha _{1}(\tau )}{\beta _{1}(\tau )},\frac{y_{1,1}}{%
\beta _{1}(\tau )},\frac{y_{1,2}}{\beta _{2}(\tau )},\frac{y_{1,3}}{\beta
_{3}(\tau )},...,\frac{y_{k+1,1}}{\beta _{1}(\tau )},\frac{y_{k+1,2}}{\beta
_{2}(\tau )},\frac{y_{k+1,3}}{\beta _{3}(\tau )}; \\
&&\frac{y_{1,1}^{\prime }}{\beta _{1}(\tau )},\frac{y_{1,2}^{\prime }}{\beta
_{2}(\tau )},\frac{y_{1,3}^{\prime }}{\beta _{3}(\tau )},...,\frac{%
y_{k+1,1}^{\prime }}{\beta _{1}(\tau )},\frac{y_{k+1,2}^{\prime }}{\beta
_{2}(\tau )},\frac{y_{k+1,3}^{\prime }}{\beta _{3}(\tau )})
\end{eqnarray*}%
if $u^{(k+1)}(t,\overrightarrow{\mathbf{x}_{k+1}};\overrightarrow{\mathbf{x}%
_{k+1}^{\prime }})$ solves%
\begin{eqnarray*}
\left( i\partial _{t}+\widetilde{L_{\overrightarrow{\mathbf{x}_{k+1}}}}(t)-%
\widetilde{L_{\overrightarrow{\mathbf{x}_{k+1}^{\prime }}}}(t)\right)
u^{(k+1)} &=&0\text{ in }\mathbb{R}^{6k+6}\times \lbrack -\upsilon
_{1}^{-1}(T_{0}),\upsilon _{1}^{-1}(T_{0})] \\
u^{(k+1)}(0,\overrightarrow{\mathbf{x}_{k+1}};\overrightarrow{\mathbf{x}%
_{k+1}^{\prime }}) &=&\gamma _{0}^{(k+1)}.
\end{eqnarray*}%
At long last, define%
\begin{equation*}
L_{\mathbf{x}}(t)={\Huge \{}%
\begin{array}{c}
\widetilde{L_{\mathbf{x}}}(t),\text{ when }t\in \lbrack -\upsilon
_{1}^{-1}(T_{0}),\upsilon _{1}^{-1}(T_{0})] \\ 
\widetilde{L_{\mathbf{x}}}(\upsilon _{1}^{-1}(T_{0})),\text{ when }%
t\geqslant \upsilon _{1}^{-1}(T_{0})\text{ or }t\leqslant -\upsilon
_{1}^{-1}(T_{0})%
\end{array}%
\end{equation*}%
then we obtain the desired variant of the generalized lens transform i.e.
Proposition \ref{Proposition:TheLensTransformNeeded}.

\section{Proof of Theorem \protect\ref{Theorem:Collapsing for GP}\label%
{Sec:ProofOfGPCollapsing}}

Without loss of generality, we show Theorem \ref{Theorem:Collapsing for GP}
for $B_{j,k+1}^{1}$ in $B_{j,k+1}$ when $j$ is taken to be $1.$ This
corresponds to the estimate:%
\begin{eqnarray}
&&\int_{s}^{T}d\tau \int_{\mathbb{R}^{3k}\times \mathbb{R}^{3k}}\left\vert
R_{\tau }^{(k)}\gamma ^{(k+1)}(\tau ,\overrightarrow{\mathbf{y}_{k}},\mathbf{%
y}_{1};\overrightarrow{\mathbf{y}_{k}^{\prime }},\mathbf{y}_{1})\right\vert
^{2}d\overrightarrow{\mathbf{y}_{k}}d\overrightarrow{\mathbf{y}_{k}^{\prime }%
}  \label{estimate:Collapse} \\
&\leqslant &C\left( \inf_{\tau \in \lbrack
0,T_{0}]}\dprod\limits_{l=2}^{3}\beta _{l}^{2}(\tau )\right) ^{-1}\int_{%
\mathbb{R}^{3(k+1)}\times \mathbb{R}^{3(k+1)}}\left\vert R_{\tau
}^{(k+1)}\gamma ^{(k+1)}(\tau ,\overrightarrow{\mathbf{y}_{k+1}};%
\overrightarrow{\mathbf{y}_{k+1}^{\prime }})\right\vert ^{2}d\overrightarrow{%
\mathbf{y}_{k+1}}d\overrightarrow{\mathbf{y}_{k+1}^{\prime }},  \notag
\end{eqnarray}%
$\forall \tau \in \lbrack s,T],$ if $\gamma ^{(k+1)}$ satisfies equation \ref%
{eqn:homogeneous hierarchy with anisotropic traps}.

By Proposition \ref{Proposition:TheLensTransformNeeded}, we compute%
\begin{eqnarray}
&&R_{\tau }^{(k)}\gamma ^{(k+1)}(\tau ,\overrightarrow{\mathbf{y}_{k}},%
\mathbf{y}_{1};\overrightarrow{\mathbf{y}_{k}^{\prime }},\mathbf{y}_{1})
\label{equation:naturality} \\
&=&\left( \dprod\limits_{l=1}^{3}\frac{1}{\beta _{l}(\tau )}\right)
\tprod_{j=1}^{k}\left( \dprod\limits_{l=1}^{3}\frac{e^{i\frac{\dot{\beta}%
_{l}(\tau )}{\beta _{l}(\tau )}\frac{\left( \left\vert y_{j,l}\right\vert
^{2}-\left\vert y_{j,l}^{\prime }\right\vert ^{2}\right) }{2}}}{\beta
_{l}(\tau )}\right) \left( \left( \prod_{j=1}^{k}\left( \nabla _{\mathbf{x}%
_{j}}\nabla _{\mathbf{x}_{j}^{\prime }}\right) \right) u^{(k+1)}(\frac{%
\alpha _{1}(\tau )}{\beta _{1}(\tau )},\overrightarrow{\mathbf{x}_{k}},%
\mathbf{x}_{1};\overrightarrow{\mathbf{x}_{k}^{\prime }},\mathbf{x}%
_{1})\right) ,  \notag
\end{eqnarray}%
if we let%
\begin{equation*}
x_{j,l}=\frac{y_{j,l}}{\beta _{l}(\tau )}\text{ and }x_{j,l}^{\prime }=\frac{%
y_{j,l}^{\prime }}{\beta _{l}(\tau )},\text{ }
\end{equation*}%
because of the relations%
\begin{equation*}
i\beta _{l}(\tau )\frac{\partial }{\partial y_{j,l}}\left( e^{i\frac{\dot{%
\beta}_{l}(\tau )}{\beta _{l}(\tau )}\frac{\left\vert y_{j,l}\right\vert ^{2}%
}{2}}\right) +\dot{\beta}_{l}(\tau )y_{j,l}\left( e^{i\frac{\dot{\beta}%
_{l}(\tau )}{\beta _{l}(\tau )}\frac{\left\vert y_{j,l}\right\vert ^{2}}{2}%
}\right) =0,
\end{equation*}%
\begin{equation*}
\beta _{l}(\tau )\frac{\partial }{\partial y_{j,l}}=\frac{\partial }{%
\partial x_{j,l}}.
\end{equation*}%
Consequently,%
\begin{eqnarray*}
&&\int_{s}^{T}d\tau \int_{\mathbb{R}^{3k}\times \mathbb{R}^{3k}}\left\vert
R_{\tau }^{(k)}\gamma ^{(k+1)}(\tau ,\overrightarrow{\mathbf{y}_{k}},\mathbf{%
y}_{1};\overrightarrow{\mathbf{y}_{k}^{\prime }},\mathbf{y}_{1})\right\vert
^{2}d\overrightarrow{\mathbf{y}_{k}}d\overrightarrow{\mathbf{y}_{k}^{\prime }%
} \\
&=&\int_{s}^{T}d\tau \int_{\mathbb{R}^{6k}}\left\vert \left(
\dprod\limits_{l=1}^{3}\frac{1}{\beta _{l}(\tau )}\right) ^{k+1}\left(
\prod_{j=1}^{k}\left( \nabla _{\mathbf{x}_{j}}\nabla _{\mathbf{x}%
_{j}^{\prime }}\right) \right) u^{(k+1)}(\frac{\alpha _{1}(\tau )}{\beta
_{1}(\tau )},\overrightarrow{\mathbf{x}_{k}},\mathbf{x}_{1};\overrightarrow{%
\mathbf{x}_{k}^{\prime }},\mathbf{x}_{1})\right\vert ^{2}d\overrightarrow{%
\mathbf{y}_{k}}d\overrightarrow{\mathbf{y}_{k}^{\prime }} \\
&=&\int_{s}^{T}\frac{d\tau }{\left( \beta _{1}(\tau )\right) ^{2}}\int_{%
\mathbb{R}^{6k}}\left( \dprod\limits_{l=2}^{3}\frac{1}{\beta _{l}(\tau )}%
\right) ^{2}\left\vert \left( \prod_{j=1}^{k}\left( \nabla _{\mathbf{x}%
_{j}}\nabla _{\mathbf{x}_{j}^{\prime }}\right) \right) u^{(k+1)}(\frac{%
\alpha _{1}(\tau )}{\beta _{1}(\tau )},\overrightarrow{\mathbf{x}_{k}},%
\mathbf{x}_{1};\overrightarrow{\mathbf{x}_{k}^{\prime }},\mathbf{x}%
_{1})\right\vert ^{2}d\overrightarrow{\mathbf{x}_{k}}d\overrightarrow{%
\mathbf{x}_{k}^{\prime }} \\
&\leqslant &\left( \inf_{\tau \in \lbrack
0,T_{0}]}\dprod\limits_{l=2}^{3}\beta _{l}^{2}(\tau )\right)
^{-1}\int_{s}^{T}\frac{d\tau }{\left( \beta _{1}(\tau )\right) ^{2}}\int_{%
\mathbb{R}^{6k}}\left\vert \left( \prod_{j=1}^{k}\left( \nabla _{\mathbf{x}%
_{j}}\nabla _{\mathbf{x}_{j}^{\prime }}\right) \right) u^{(k+1)}(\frac{%
\alpha _{1}(\tau )}{\beta _{1}(\tau )},\overrightarrow{\mathbf{x}_{k}},%
\mathbf{x}_{1};\overrightarrow{\mathbf{x}_{k}^{\prime }},\mathbf{x}%
_{1})\right\vert ^{2}d\overrightarrow{\mathbf{x}_{k}}d\overrightarrow{%
\mathbf{x}_{k}^{\prime }} \\
&\leqslant &\left( \inf_{\tau \in \lbrack
0,T_{0}]}\dprod\limits_{l=2}^{3}\beta _{l}^{2}(\tau )\right)
^{-1}\int_{-\infty }^{\infty }dt\int_{\mathbb{R}^{3k}\times \mathbb{R}%
^{3k}}\left\vert \left( \prod_{j=1}^{k}\left( \nabla _{\mathbf{x}_{j}}\nabla
_{\mathbf{x}_{j}^{\prime }}\right) \right) u^{(k+1)}(t,\overrightarrow{%
\mathbf{x}_{k}},\mathbf{x}_{1};\overrightarrow{\mathbf{x}_{k}^{\prime }},%
\mathbf{x}_{1})\right\vert ^{2}d\overrightarrow{\mathbf{x}_{k}}d%
\overrightarrow{\mathbf{x}_{k}^{\prime }}
\end{eqnarray*}%
where we used the fact that the Wronskian of $\alpha _{l}$ and $\beta _{l}$
is constant $1$, i.e. 
\begin{equation*}
\frac{dt}{d\tau }=\frac{\dot{\alpha}_{1}(\tau )\beta _{1}(\tau )-\alpha
_{1}(\tau )\dot{\beta}_{1}(\tau )}{\left( \beta _{1}(\tau )\right) ^{2}}=%
\frac{1}{\left( \beta _{1}(\tau )\right) ^{2}}
\end{equation*}%
as shown in Claim \ref{Claim:Properties of Alpha and Beta}.

A corollary of Theorem \ref{Theorem:3*nd} tells us that

\begin{corollary}
Let $L_{\mathbf{x}}(t)$ be the same as in Theorem \ref{Theorem:3*nd} and $%
u^{(k+1)}$ verify%
\begin{equation*}
\left( i\partial _{t}+L_{\overrightarrow{\mathbf{x}_{k+1}}}(t)-L_{%
\overrightarrow{\mathbf{x}_{k+1}^{\prime }}}(t)\right) u^{(k+1)}=0.
\end{equation*}%
Then there is a $C>0,$ independent of $j,$ $k,$ and $u^{(k+1)}$ s.t. 
\begin{eqnarray*}
&&\left\Vert \left( \prod_{j=1}^{k}\left( \nabla _{\mathbf{x}_{j}}\nabla _{%
\mathbf{x}_{j}^{\prime }}\right) \right) \left(
B_{j,k+1}^{1}u^{(k+1)}\right) (t,\overrightarrow{\mathbf{x}_{k}};%
\overrightarrow{\mathbf{x}_{k}^{\prime }})\right\Vert _{L^{2}(\mathbb{R}%
\times \mathbb{R}^{3k}\times \mathbb{R}^{3k})} \\
&=&\left\Vert \left( \prod_{j=1}^{k}\left( \nabla _{\mathbf{x}_{j}}\nabla _{%
\mathbf{x}_{j}^{\prime }}\right) \right) u^{(k+1)}(t,\overrightarrow{\mathbf{%
x}_{k}},\mathbf{x}_{1};\overrightarrow{\mathbf{x}_{k}^{\prime }},\mathbf{x}%
_{1})\right\Vert _{L^{2}(\mathbb{R}\times \mathbb{R}^{3k}\times \mathbb{R}%
^{3k})} \\
&\leqslant &C\left\Vert \left( \prod_{j=1}^{k+1}\left( \nabla _{\mathbf{x}%
_{j}}\nabla _{\mathbf{x}_{j}^{\prime }}\right) \right) u^{(k+1)}(0,%
\overrightarrow{\mathbf{x}_{k+1}};\overrightarrow{\mathbf{x}_{k+1}^{\prime }}%
)\right\Vert _{L^{2}(\mathbb{R}^{3(k+1)}\times \mathbb{R}^{3(k+1)})},
\end{eqnarray*}
\end{corollary}

Whence inequality \ref{estimate:Collapse} follows.

\section{The Uniqueness of Hierarchy \protect\ref{equation:Gross-Pitaevskii
hiearchy with anisotropic traps}\label{Sec:ProofOfUniqueness}}

To get Theorem \ref{Theorem:Uniqueness of GP}, we of course use the
Klainerman-Machedon board game argument to group the terms. For convenience,
we assume $b_{0}=1$ here.

\begin{lemma}
\label{lemma:MateiLemma}One can express $\gamma ^{(1)}(\tau _{1},\mathbf{%
\cdot };\mathbf{\cdot })$ in the Gross-Pitaevskii hierarchy \ref%
{equation:Gross-Pitaevskii hiearchy with anisotropic traps} as a sum of at
most $4^{n}$ terms of the form 
\begin{equation*}
\int_{D}J(\underline{\tau }_{n+1},\mu _{m})d\underline{\tau }_{n+1},
\end{equation*}%
or in other words, 
\begin{equation}
\gamma ^{(1)}(\tau _{1},\mathbf{\cdot };\mathbf{\cdot })=\sum_{m}\int_{D}J(%
\underline{\tau }_{n+1},\mu _{m})d\underline{\tau }_{n+1}.
\label{formula:MateiKlainermanFormula}
\end{equation}%
Here $\underline{\tau }_{n+1}=(\tau _{2},\tau _{3},...,\tau _{n+1})$, $%
D\subset \lbrack s,\tau _{1}]^{n}$, $\mu _{m}$ are a set of maps from $%
\{2,...,n+1\}$ to $\{1,...,n\}$ satisfying $\mu _{m}(2)=1$ and $\mu
_{m}(j)<j $ for all $j,$ and%
\begin{eqnarray*}
J(\underline{\tau }_{n+1},\mu _{m}) &=&U^{(1)}(\tau _{1};\tau
_{2})B_{1,2}U^{(2)}(\tau _{2};\tau _{3})B_{\mu _{m}(3),2}... \\
&&U^{(n)}(\tau _{n};\tau _{n+1})B_{\mu _{m}(n+1),n+1}(\gamma ^{(n+1)}(\tau
_{n+1},\mathbf{\cdot };\mathbf{\cdot })).
\end{eqnarray*}
\end{lemma}

\begin{proof}
The RHS of formula \ref{formula:MateiKlainermanFormula} is in fact a Duhamel
principle. This lemma follows from the proof of Theorem 3.4 in \cite%
{KlainermanAndMachedon} which uses a board game inspired by the Feynman
graph argument in \cite{E-S-Y2}. One just needs to replace $%
e^{i(t_{1}-t_{2})\triangle _{y}}$ by $U_{\mathbf{y}}(t_{1};t_{2})$, and $%
e^{i(t_{1}-t_{2})\triangle ^{(k)}}$ by $U^{(k)}(t_{1};t_{2}).$
\end{proof}

Let $D_{\tau _{2}}=\left\{ \left( \tau _{3},...,\tau _{n+1}\right) |\left(
\tau _{2},\tau _{3},...,\tau _{n+1}\right) \in D\right\} $ where $D$ is as
in Lemma \ref{lemma:MateiLemma}. Assuming that we have already verified 
\begin{equation*}
\left\Vert R_{s}^{(1)}\gamma ^{(1)}(s,\cdot )\right\Vert _{L^{2}(\mathbb{R}%
^{3}\times \mathbb{R}^{3})}=0,
\end{equation*}%
applying Lemma \ref{lemma:MateiLemma} to $[s,\tau _{1}]\subset \lbrack
0,T_{0}]$, we have%
\begin{eqnarray*}
&&\left\Vert R_{\tau _{1}}^{(1)}\gamma ^{(1)}(\tau _{1},\cdot )\right\Vert
_{L^{2}(\mathbb{R}^{3}\times \mathbb{R}^{3})} \\
&=&\left\Vert R_{\tau _{1}}^{(1)}\int_{D}U^{(1)}(\tau _{1};\tau
_{2})B_{1,2}U^{(2)}(\tau _{2};\tau _{3})B_{\mu _{m}(3),2}...d\tau
_{2}...d\tau _{n+1}\right\Vert _{L^{2}(\mathbb{R}^{3}\times \mathbb{R}^{3})}
\\
&=&\left\Vert \int_{s}^{\tau _{1}}U^{(1)}(\tau _{1};\tau _{2})\left(
\int_{D_{\tau _{2}}}R_{\tau _{2}}^{(1)}B_{1,2}U^{(2)}(\tau _{2};\tau
_{3})B_{\mu _{m}(3),2}...d\tau _{3}...d\tau _{n+1}\right) d\tau
_{2}\right\Vert _{L^{2}(\mathbb{R}^{3}\times \mathbb{R}^{3})}\text{ } \\
&&\text{(Lemma \ref{Lemma:Monentum})} \\
&\leqslant &\int_{s}^{\tau _{1}}\left\Vert \int_{D_{\tau _{2}}}R_{\tau
_{2}}^{(1)}B_{1,2}U^{(2)}(\tau _{2};\tau _{3})B_{\mu _{m}(3),2}...d\tau
_{3}...d\tau _{n+1}\right\Vert _{L^{2}(\mathbb{R}^{3}\times \mathbb{R}%
^{3})}d\tau _{2} \\
&\leqslant &\int_{[s,\tau _{1}]^{n}}\left\Vert R_{\tau
_{2}}^{(1)}B_{1,2}U^{(2)}(\tau _{2};\tau _{3})B_{\mu
_{m}(3),2}...\right\Vert _{L^{2}(\mathbb{R}^{3}\times \mathbb{R}^{3})}d\tau
_{2}d\tau _{3}...d\tau _{n+1} \\
&\leqslant &\left( \tau _{1}-s\right) ^{\frac{1}{2}}\int_{[s,\tau
_{1}]^{n-1}}\left\Vert R_{\tau _{2}}^{(1)}B_{1,2}U^{(2)}(\tau _{2};\tau
_{3})B_{\mu _{m}(3),2}...\right\Vert _{L^{2}(\tau _{2}\in \lbrack s,\tau
_{1}]\times \mathbb{R}^{3}\times \mathbb{R}^{3})}d\tau _{3}...d\tau _{n+1} \\
&\leqslant &C\left( \tau _{1}-s\right) ^{\frac{1}{2}}\int_{[s,\tau
_{1}]^{n-1}}\left\Vert R_{\tau _{2}}^{(2)}U^{(2)}(\tau _{2};\tau _{3})B_{\mu
_{m}(3),2}...\right\Vert _{L^{2}(\mathbb{R}^{6}\times \mathbb{R}^{6})}d\tau
_{3}...d\tau _{n+1}\text{ }\left( \text{Theorem \ref{Theorem:Collapsing for
GP}}\right) \\
&&(\text{Same procedure }n-2\text{ times}) \\
&\leqslant &C\left( C\left( \tau _{1}-s\right) \right) ^{\frac{n-1}{2}%
}\int_{s}^{\tau _{1}}\left\Vert R_{\tau _{n+1}}^{(n)}B_{\mu
_{m}(n+1),n+1}\gamma ^{(n+1)}(\tau _{n+1},\cdot )\right\Vert _{L^{2}(\mathbb{%
R}^{3n}\times \mathbb{R}^{3n})}d\tau _{n+1} \\
&\leqslant &C\left( C\left( \tau _{1}-s\right) \right) ^{\frac{n-1}{2}}.
\end{eqnarray*}

Let $\left( \tau _{1}-s\right) $ be sufficiently small, and $n\rightarrow
\infty $, we infer that 
\begin{equation*}
\left\Vert R_{\tau _{1}}^{(1)}\gamma ^{(1)}(\tau _{1},\cdot )\right\Vert
_{L^{2}(\mathbb{R}^{3}\times \mathbb{R}^{3})}=0\text{ in }[s,\tau _{1}].
\end{equation*}%
Similar arguments show that $\left\Vert R_{\tau }^{(k)}\gamma ^{(k)}(\tau
,\cdot )\right\Vert _{L^{2}(\mathbb{R}^{3}\times \mathbb{R}^{3})}=0$, $%
\forall k,\tau \in \lbrack 0,T_{0}].$ Hence we have attained Theorem \ref%
{Theorem:Uniqueness of GP}.

\section{Derivation of the 2d Cubic NLS with Anisotropic Switchable
Quadratic Traps / Proof of Theorem \protect\ref{Theorem:BECin2D}\label%
{Sec:ProofOfBEC}}

For a more comprehensible presentation, let us suppose%
\begin{equation*}
H_{\mathbf{y}}(\tau )=\sum_{l=1}^{n}\left( -\frac{\partial ^{2}}{\partial
y_{j,l}^{2}}+\eta _{l}(\tau )y_{j,l}^{2}\right)
\end{equation*}%
is the ordinary Hermite operator%
\begin{equation*}
H_{\mathbf{y}}=-\triangle _{\mathbf{y}}+\left\vert \mathbf{y}\right\vert ^{2}
\end{equation*}
in this section to make formulas shorter and more explicit. We will add two
remarks in the proof to address the small modifications needed for the
general case.

We start by reviewing the standard Elgart-Erd\"{o}s-Schlein-Yau program in
this setting.

\begin{itemize}
\item[Step A.] Observe that, by definition, $\left\{ \gamma
_{N}^{(k)}\right\} $ solves the quadratic trap
Bogoliubov--Born--Green--Kirkwood--Yvon (BBGKY) hierarchy 
\begin{eqnarray}
&&\left( i\partial _{\tau }-\frac{1}{2}\left( -\triangle _{\overrightarrow{%
\mathbf{y}_{k}}}+\left\vert \overrightarrow{\mathbf{y}_{k}}\right\vert
^{2}\right) +\frac{1}{2}\left( -\triangle _{\overrightarrow{\mathbf{y}%
_{k}^{\prime }}}+\left\vert \overrightarrow{\mathbf{y}_{k}^{\prime }}%
\right\vert ^{2}\right) \right) \gamma _{N}^{(k)}  \label{hierarchy:BBGKY} \\
&=&\frac{1}{N}\sum_{1\leqslant i<j\leqslant k}\left( V_{N}(\mathbf{y}_{i}-%
\mathbf{y}_{j})-V_{N}(\mathbf{y}_{i}^{\prime }-\mathbf{y}_{j}^{\prime
})\right) \gamma _{N}^{(k)}  \notag \\
&&+\frac{N-k}{N}\sum_{j=1}^{k}\int dy_{k+1}[\left( V_{N}(\mathbf{y}_{i}-%
\mathbf{y}_{k+1})-V_{N}(\mathbf{y}_{i}^{\prime }-\mathbf{y}_{k+1})\right) 
\notag \\
&&\gamma _{N}^{(k+1)}(\tau ,\overrightarrow{\mathbf{y}_{k}},\mathbf{y}_{k+1};%
\overrightarrow{\mathbf{y}_{k}^{\prime }},\mathbf{y}_{k+1})]  \notag
\end{eqnarray}%
where $V_{N}(\mathbf{x})=N^{n\beta }V\left( N^{\beta }\mathbf{x}\right) $.
It converges (at least formally) to the quadratic trap Gross-Pitaevskii
infinite hierarchy%
\begin{eqnarray}
&&\left( i\partial _{\tau }-\frac{1}{2}\left( -\triangle _{\overrightarrow{%
\mathbf{y}_{k}}}+\left\vert \overrightarrow{\mathbf{y}_{k}}\right\vert
^{2}\right) +\frac{1}{2}\left( -\triangle _{\overrightarrow{\mathbf{y}%
_{k}^{\prime }}}+\left\vert \overrightarrow{\mathbf{y}_{k}^{\prime }}%
\right\vert ^{2}\right) \right) \gamma ^{(k)}  \label{hierarchy:Q-GP} \\
&=&b_{0}\sum_{j=1}^{k}B_{j,k+1}\left( \gamma ^{(k+1)}\right) .  \notag
\end{eqnarray}%
Prove rigorously that the sequence $\left\{ \gamma _{N}^{(k)}\right\} $ is
compact with respect to the weak* topology on the trace class operators and
every limit point $\left\{ \gamma ^{(k)}\right\} $ satisfies hierarchy \ref%
{hierarchy:Q-GP}.

\item[Step B.] Utilize a suitable uniqueness theorem of hierarchy \ref%
{hierarchy:Q-GP} to conclude that 
\begin{equation*}
\gamma ^{(k)}(\tau ,\overrightarrow{\mathbf{y}_{k}};\overrightarrow{\mathbf{y%
}_{k}^{\prime }})=\dprod\limits_{j=1}^{k}\phi (\tau ,\mathbf{y}_{j})%
\overline{\phi (\tau ,\mathbf{y}_{j}^{\prime })},
\end{equation*}%
where $\phi $ solves the 2d quadratic trap cubic NLS%
\begin{equation*}
i\partial _{\tau }\phi =\frac{1}{2}\left( -\triangle +\left\vert \mathbf{y}%
\right\vert ^{2}\right) \phi +b_{0}\phi \left\vert \phi \right\vert ^{2}.
\end{equation*}%
So the compact sequence $\left\{ \gamma _{N}^{(k)}\right\} $ has only one
limit point, i.e. 
\begin{equation*}
\gamma _{N}^{(k)}\rightarrow \dprod\limits_{j=1}^{k}\phi (\tau ,\mathbf{y}%
_{j})\overline{\phi (\tau ,\mathbf{y}_{j}^{\prime })}
\end{equation*}%
in the weak* topology. Since $\gamma ^{(k)}$ is an orthogonal projection,
the convergence in the weak* topology is equivalent to the convergence in
the trace norm topology.
\end{itemize}

We modify this procedure to show Theorem \ref{Theorem:BECin2D}. We remark
that the main additional tool is the lens transform. When $H_{\mathbf{y}%
}(\tau )$ is the Hermite operator, $\alpha _{l}=\sin \tau $, $\beta
_{l}=\cos \tau $ and $T_{0}<\frac{\pi }{2}$ i.e. the lens transform and its
inverse reads as follow.

\begin{definition}
We define the lens transform $T_{l}:$ $L^{2}(d\overrightarrow{\mathbf{x}_{k}}%
d\overrightarrow{\mathbf{x}_{k}^{\prime }})\rightarrow $ $L^{2}(d%
\overrightarrow{\mathbf{y}_{k}}d\overrightarrow{\mathbf{y}_{k}^{\prime }})$
and its inverse by%
\begin{eqnarray*}
\left( T_{l}u^{(k)}\right) (\tau ,\overrightarrow{\mathbf{y}_{k}};%
\overrightarrow{\mathbf{y}_{k}^{\prime }}) &=&\frac{1}{(\cos \tau )^{nk}}%
u^{(k)}(\tan \tau ,\frac{\overrightarrow{\mathbf{y}_{k}}}{\cos \tau };\frac{%
\overrightarrow{\mathbf{y}_{k}^{\prime }}}{\cos \tau })e^{-i\frac{\tan \tau 
}{2}(\left\vert \overrightarrow{\mathbf{y}_{k}}\right\vert ^{2}-\left\vert 
\overrightarrow{\mathbf{y}_{k}^{\prime }}\right\vert ^{2})} \\
\left( T_{l}^{-1}\gamma ^{(k)}\right) (t,\overrightarrow{\mathbf{x}_{k}};%
\overrightarrow{\mathbf{x}_{k}^{\prime }}) &=&\frac{1}{\left( 1+t^{2}\right)
^{\frac{nk}{2}}}\gamma ^{(k)}(\arctan t,\frac{\overrightarrow{\mathbf{x}_{k}}%
}{\sqrt{1+t^{2}}};\frac{\overrightarrow{\mathbf{x}_{k}^{\prime }}}{\sqrt{%
1+t^{2}}})e^{\frac{it}{2(1+t^{2})}\left( \left\vert \overrightarrow{\mathbf{x%
}_{k}}\right\vert ^{2}-\left\vert \overrightarrow{\mathbf{x}_{k}^{\prime }}%
\right\vert ^{2}\right) }.
\end{eqnarray*}%
$T_{l}$ is unitary by definition and the variables are related by%
\begin{equation*}
\tau =\arctan t,\text{ }\mathbf{y}_{k}=\frac{\mathbf{x}_{k}}{\sqrt{1+t^{2}}}%
\text{and }\mathbf{y}_{k}^{\prime }=\frac{\mathbf{x}_{k}^{\prime }}{\sqrt{%
1+t^{2}}}.
\end{equation*}
\end{definition}

\begin{remark}
For the general anisotropic case, we still need the 2d version of
Proposition \ref{Proposition:TheLensTransformNeeded}.
\end{remark}

Let us write%
\begin{equation*}
\left( T_{l}^{-1}\gamma ^{(k)}\right) (t,\overrightarrow{\mathbf{x}_{k}};%
\overrightarrow{\mathbf{x}_{k}^{\prime }})=\gamma ^{(k)}(\tau ,%
\overrightarrow{\mathbf{y}_{k}};\overrightarrow{\mathbf{y}_{k}^{\prime }})%
\frac{e^{\frac{it}{2(1+t^{2})}\left( \left\vert \mathbf{x}_{k}\right\vert
^{2}-\left\vert \mathbf{x}_{k}^{\prime }\right\vert ^{2}\right) }}{\left(
1+t^{2}\right) ^{\frac{nk}{2}}}:=\gamma ^{(k)}(\tau ,\overrightarrow{\mathbf{%
y}_{k}};\overrightarrow{\mathbf{y}_{k}^{\prime }})h_{n}^{(k)}(t,%
\overrightarrow{\mathbf{x}_{k}};\overrightarrow{\mathbf{x}_{k}^{\prime }}),
\end{equation*}%
then we have a more explicit version of Proposition \ref%
{Proposition:TheLensTransformNeeded}.

\begin{proposition}
\label{Proposition:LensTransformRelation}%
\begin{eqnarray*}
&&\left( i\partial _{t}+\frac{1}{2}\triangle _{\overrightarrow{\mathbf{x}_{k}%
}}-\frac{1}{2}\triangle _{\overrightarrow{\mathbf{x}_{k}^{\prime }}}\right)
\left( T_{l}^{-1}\gamma ^{(k)}\right) (t,\overrightarrow{\mathbf{x}_{k}};%
\overrightarrow{\mathbf{x}_{k}^{\prime }}) \\
&=&\frac{h_{n}^{(k)}}{1+t^{2}}\left[ \left( i\partial _{\tau }-\frac{1}{2}%
\left( -\triangle _{\overrightarrow{\mathbf{y}_{k}}}+\left\vert 
\overrightarrow{\mathbf{y}_{k}}\right\vert ^{2}\right) +\frac{1}{2}\left(
-\triangle _{\overrightarrow{\mathbf{y}_{k}^{\prime }}}+\left\vert 
\overrightarrow{\mathbf{y}_{k}^{\prime }}\right\vert ^{2}\right) \right)
\gamma ^{(k)}(\tau ,\overrightarrow{\mathbf{y}_{k}};\overrightarrow{\mathbf{y%
}_{k}^{\prime }})\right]
\end{eqnarray*}
\end{proposition}

\begin{proof}
This is a direct computation.
\end{proof}

Via this proposition, we understand how the lens transform acts on
hierarchies \ref{hierarchy:BBGKY} and \ref{hierarchy:Q-GP}.

\begin{lemma}
\label{Lemma:Q-GPUnderLens}(Gross-Pitaevskii hierarchy under the lens
transform) $\left\{ \gamma ^{(k)}\right\} $ solves the quadratic trap
Gross-Pitaevskii hierarchy \ref{hierarchy:Q-GP}\ if and only if $\left\{
u^{(k)}=T_{l}^{-1}\gamma ^{(k)}\right\} $ solves the infinite hierarchy%
\begin{equation}
\left( i\partial _{t}+\frac{1}{2}\triangle _{\overrightarrow{\mathbf{x}_{k}}%
}-\frac{1}{2}\triangle _{\overrightarrow{\mathbf{x}_{k}^{\prime }}}\right)
u^{(k)}=\frac{\left( 1+t^{2}\right) ^{\frac{n}{2}}}{1+t^{2}}%
b_{0}\sum_{j=1}^{k}B_{j,k+1}\left( u^{(k+1)}\right) .
\label{hierarchy:Q-GP after Lens transform}
\end{equation}%
In particular, when $n=2$, the lens transform preserves the Gross-Pitaevskii
hierarchy.
\end{lemma}

\begin{lemma}
(BBGKY hierarchy under the lens transform) $\left\{ \gamma
_{N}^{(k)}\right\} $ solves the quadratic trap BBGKY hierarchy \ref%
{hierarchy:BBGKY}\ if and only if $\left\{ u_{N}^{(k)}=T_{l}^{-1}\gamma
_{N}^{(k)}\right\} $ solves the hierarchy%
\begin{eqnarray}
&&\left( i\partial _{t}+\frac{1}{2}\triangle _{\overrightarrow{\mathbf{x}_{k}%
}}-\frac{1}{2}\triangle _{\overrightarrow{\mathbf{x}_{k}^{\prime }}}\right)
u_{N}^{(k)}  \label{hiearchy:BBGKY-T} \\
&=&\frac{1}{N}\frac{1}{1+t^{2}}\sum_{1\leqslant i<j\leqslant k}\left( V_{N}(%
\frac{\mathbf{x}_{i}-\mathbf{x}_{j}}{\sqrt{1+t^{2}}})-V_{N}(\frac{\mathbf{x}%
_{i}^{\prime }-\mathbf{x}_{j}^{\prime }}{\sqrt{1+t^{2}}})\right) u_{N}^{(k)}
\notag \\
&&+\frac{N-k}{N}\frac{1}{1+t^{2}}\sum_{j=1}^{k}\int d\mathbf{x}_{k+1}[\left(
V_{N}(\frac{\mathbf{x}_{i}-\mathbf{x}_{k+1}}{\sqrt{1+t^{2}}})-V_{N}(\frac{%
\mathbf{x}_{i}^{\prime }-\mathbf{x}_{k+1}}{\sqrt{1+t^{2}}})\right)  \notag \\
&&u_{N}^{(k+1)}(t,\overrightarrow{\mathbf{x}_{k}},\mathbf{x}_{k+1};%
\overrightarrow{\mathbf{x}_{k}^{\prime }},\mathbf{x}_{k+1})],  \notag
\end{eqnarray}
\end{lemma}

We can now prove Theorem \ref{Theorem:BECin2D}.

\subsection{Proof of Theorem \protect\ref{Theorem:BECin2D}}

\begin{itemize}
\item[Step 1.] Let $n=2$, consider $\left\{ u_{N}^{(k)}=T_{l}^{-1}\gamma
_{N}^{(k)}\right\} $ which solves hierarchy \ref{hiearchy:BBGKY-T}.

\item[Step 2.] Write%
\begin{equation*}
\widetilde{V}(\mathbf{x})=\frac{1}{1+t^{2}}V(\frac{\mathbf{x}}{\sqrt{1+t^{2}}%
}),
\end{equation*}%
then%
\begin{equation*}
\frac{1}{\left( 1+T^{2}\right) ^{1-\frac{1}{p}}}\left\Vert V\right\Vert
_{p}\leqslant \left\Vert \widetilde{V}\right\Vert _{p}\leqslant \left\Vert
V\right\Vert _{p}\text{ when }T<\infty \text{ and }p\geqslant 1.
\end{equation*}%
Therefore we can employ the proof in Kirkpatrick-Schlein-Staffilani \cite%
{Kirpatrick} to show that the sequence $\left\{ u_{N}^{(k)}\right\} $ is
compact with respect to the weak* topology on the trace class operators and
every limit point $\left\{ u^{(k)}\right\} $ satisfies the Gross-Pitaevskii
hierarchy \ref{hierarchy:Q-GP after Lens transform}. Moreover, based on a
fixed time trace theorem argument as in \cite{Kirpatrick}, for $\alpha <1$,
we have 
\begin{equation*}
\int_{0}^{T}dt\left\Vert \dprod\limits_{j=1}^{k}\left( \left\langle \nabla _{%
\mathbf{x}_{j}}\right\rangle ^{\alpha }\left\langle \nabla _{\mathbf{x}%
_{j}}\right\rangle ^{\alpha }\right) B_{j,k+1}\left( u^{(k+1)}\right)
\right\Vert _{L^{2}(\mathbb{R}^{2k}\times \mathbb{R}^{2k})}\leqslant C^{k}.
\end{equation*}%
for every limit point $\left\{ u^{(k)}\right\} $. To be more precise, the
proof in \cite{Kirpatrick} involves a smooth approximation. We omit this
detail here.
\end{itemize}

\begin{remark}
The auxiliary Hamiltonian 
\begin{equation*}
\widetilde{H_{N}}(t)=\frac{1}{2}\sum_{j=1}^{N}L_{\mathbf{X}_{j}}(t)+\frac{1}{%
N}\sum_{i<j}N^{2\beta }\tilde{V}(N^{\beta }\left( \mathbf{x}_{i}-\mathbf{x}%
_{j}\right) ).
\end{equation*}%
which corresponds to the anisotropic quadratic potential case does not lead
to the conservation of the quantity%
\begin{equation*}
\left\langle \psi _{N},\left( \widetilde{H_{N}}(t)\right) ^{k}\psi
_{N}\right\rangle .
\end{equation*}%
On the other hand, the following estimate controls the energy.%
\begin{eqnarray*}
\frac{d}{dt}\left\langle \psi _{N},\left( \widetilde{H_{N}}(t)\right)
^{k}\psi _{N}\right\rangle &=&\left\langle \psi _{N},\left[ \frac{d}{dt}%
,\left( \widetilde{H_{N}}(t)\right) \right] \left( \widetilde{H_{N}}%
(t)\right) ^{k-1}\psi _{N}\right\rangle +... \\
&&+\left\langle \psi _{N},\left( \widetilde{H_{N}}(t)\right) ^{k-1}\left[ 
\frac{d}{dt},\left( \widetilde{H_{N}}(t)\right) \right] \psi
_{N}\right\rangle \\
&\leqslant &Ck\left\langle \psi _{N},\left( \widetilde{H_{N}}(t)\right)
^{k}\psi _{N}\right\rangle
\end{eqnarray*}%
since $a_{1}$ and $a_{2}$, the coefficients of $L_{\mathbf{X}}$, are $C^{1}$
in the context of Theorem \ref{Theorem:BECin2D}. Thus Gronwall's inequality
takes care of the problem for us as long as we are considering finite time.
\end{remark}

\begin{itemize}
\item[Step 3.] By Theorem \ref{Theorem:Uniqueness of 2d unknown GP} (2d
uniqueness) or Theorem 7.1 in \cite{Kirpatrick}, we deduce that%
\begin{equation*}
u^{(k)}(t,\overrightarrow{\mathbf{x}_{k}};\overrightarrow{\mathbf{x}%
_{k}^{\prime }})=\dprod\limits_{j=1}^{k}\tilde{\phi}(t,\mathbf{x}_{j})%
\overline{\tilde{\phi}(t,\mathbf{x}_{j}^{\prime })}
\end{equation*}%
where $\tilde{\phi}$ solves the 2d cubic NLS%
\begin{equation*}
i\partial _{t}\tilde{\phi}=-\frac{1}{2}\triangle \tilde{\phi}+b_{0}\tilde{%
\phi}\left\vert \tilde{\phi}\right\vert ^{2}.
\end{equation*}%
Hence the compact sequence $\left\{ u_{N}^{(k)}\right\} $ has only one limit
point, so 
\begin{equation*}
u_{N}^{(k)}\rightarrow \dprod\limits_{j=1}^{k}\tilde{\phi}(t,\mathbf{x}_{j})%
\overline{\tilde{\phi}(t,\mathbf{x}_{j}^{\prime })}
\end{equation*}%
in the weak* topology. Since $u^{(k)}$ is an orthogonal projection, the
convergence in the weak* topology is equivalent to the convergence in the
trace norm topology.
\end{itemize}

\begin{remark}
It is necessary to use Theorem \ref{Theorem:Uniqueness of 2d unknown GP} in
this paper for the general anisotropic quadratic traps case.
\end{remark}

\begin{itemize}
\item[Step 4.] Let $\phi $ solve the 2d quadratic trap cubic NLS%
\begin{equation*}
i\partial _{\tau }\phi =\frac{1}{2}\left( -\triangle +\left\vert \mathbf{y}%
\right\vert ^{2}\right) \phi +b_{0}\phi \left\vert \phi \right\vert ^{2},
\end{equation*}%
then the lens transform of $u^{(k)}$ is%
\begin{equation*}
\gamma ^{(k)}(\tau ,\overrightarrow{\mathbf{y}_{k}};\overrightarrow{\mathbf{y%
}_{k}^{\prime }})=\dprod\limits_{j=1}^{k}\phi (\tau ,\mathbf{y}_{j})%
\overline{\phi (\tau ,\mathbf{y}_{j}^{\prime })},
\end{equation*}%
due to the fact that the lens transform preserves mass critical NLS, which
is the cubic NLS in 2d.

\item[Step 5.] The convergence 
\begin{equation*}
u_{N}^{(k)}\rightarrow u^{(k)}
\end{equation*}%
in the trace norm indicates the convergence in the Hilbert-Schmidt norm. But
the lens transform 
\begin{equation*}
T_{l}:L^{2}(d\overrightarrow{\mathbf{x}}d\overrightarrow{\mathbf{x}%
^{^{\prime }}})\rightarrow L^{2}(d\overrightarrow{\mathbf{y}}d%
\overrightarrow{\mathbf{y}^{^{\prime }}})
\end{equation*}%
is unitary (so preserves the norm) and thus 
\begin{equation*}
\gamma _{N}^{(k)}=T_{l}u_{N}^{(k)}\rightarrow T_{l}u^{(k)}=\gamma ^{(k)}.
\end{equation*}%
Thence we conclude that $\gamma _{N}^{(k)}$ converges to%
\begin{equation*}
\gamma ^{(k)}(\tau ,\overrightarrow{\mathbf{y}_{k}};\overrightarrow{\mathbf{y%
}_{k}^{\prime }})=\dprod\limits_{j=1}^{k}\phi (\tau ,\mathbf{y}_{j})%
\overline{\phi (\tau ,\mathbf{y}_{j}^{\prime })},
\end{equation*}%
in the Hilbert-Schmidt norm, which is Theorem \ref{Theorem:BECin2D}.
\end{itemize}

\subsection{Comments about the 3d case\label{Section:3d?}}

It is natural to wonder what we can say about the 3d case using the above
method. Visiting Lemma \ref{Lemma:Q-GPUnderLens} again yields the hierarchy%
\begin{equation}
\left( i\partial _{t}+\frac{1}{2}\triangle _{\overrightarrow{\mathbf{x}_{k}}%
}-\frac{1}{2}\triangle _{\overrightarrow{\mathbf{x}_{k}^{\prime }}}\right)
u^{(k)}=\left( 1+t^{2}\right) ^{\frac{1}{2}}b_{0}\sum_{j=1}^{k}B_{j,k+1}%
\left( u^{(k+1)}\right) .  \label{hierarchy:n=3}
\end{equation}%
Due to the factor $\left( 1+t^{2}\right) ^{\frac{1}{2}}$, it is difficult to
see of what use a 3d version of Theorem \ref{Theorem:Uniqueness of 2d
unknown GP} might be. We can certainly give a uniqueness theorem regarding
hierarchy \ref{hierarchy:n=3} with the techniques in this paper. But it is
unknown how to verify the space-time bound when $n=3$ as stated earlier,

Another possibility to attack the 3d case is the standard
Elgart-Erdos-Schlein-Yau procedure, but we presently know very little about
the analysis of the Hermite like operator $H_{\mathbf{y}}(\tau )$.

Finally, we remark that it is not clear whether the Feynman diagrams
argument, the key to the uniqueness theorem in \cite{E-S-Y2} on which \cite%
{E-S-Y1,E-S-Y2,E-S-Y4, E-S-Y5, E-S-Y3} are based, leads to a 3d uniqueness
theorem of hierarchy \ref{equation:Gross-Pitaevskii hiearchy with
anisotropic traps} or \ref{hierarchy:n=3}, which represent the two sides of
the lens transform.

\section{Conclusion}

In this paper, we have derived rigorously the 2d cubic NLS with anisotropic
switchable quadratic traps through a modified Elgart-Erd\"{o}s-Schlein-Yau
procedure. We have attained partial results in 3d as well. Unfortunately,
when $n=3$, we still have unsolved problems as stated in Section \ref%
{Section:3d?}.

\section{Appendix: the Generalized Lens Transform and the Metaplectic
Representation}

In this appendix, we prove Lemmas \ref{Lemma:Monentum} and \ref%
{Lemma:GLensTransform} via the metaplectic representation. The 3d
anisotropic case drops out once we show the 1d case. Before we delve into
the proof, we remark that we currently do not have an explanation away from
direct computations for Proposition \ref{Proposition:LensTransformRelation}
or for the fact that the generalized lens transform preserves $L^{2}$
critical NLS. The group theory proof presented in this appendix only shows
the linear case: Lemmas \ref{Lemma:Monentum} and \ref{Lemma:GLensTransform}.

Through out this appendix, we consider the metaplectic representation%
\begin{equation*}
\mu :Sp\left( 2,\mathbb{R}\right) \rightarrow Unitary\text{ }Operators\text{ 
}on\text{ }L^{2}(\mathbb{R}).
\end{equation*}%
which has the property:%
\begin{equation*}
d\mu \left( 
\begin{pmatrix}
0 & 1 \\ 
-\eta (\tau ) & 0%
\end{pmatrix}%
\right) =i\left( -\frac{1}{2}\partial _{y}^{2}+\eta (\tau )\frac{y^{2}}{2}%
\right) .
\end{equation*}%
For more information regarding $\mu $ and $d\mu $, we refer the readers to
Folland's monograph \cite{Folland}. We comment that $\mu $ is not a
well-defined group homomorphism on all of $Sp\left( 2,\mathbb{R}\right) ,$
but the fact that it is well-defined in a neighborhood of the identity of $%
Sp\left( 2,\mathbb{R}\right) $ is good enough for our purpose here.

\subsection{Proof of Lemma \protect\ref{Lemma:GLensTransform} / the
Generalized Lens Transform}

\begin{proposition}
Define $\alpha $ and $\beta $ through the system%
\begin{eqnarray*}
\ddot{\alpha}(\tau )+\eta (\tau )\alpha (\tau ) &=&0,\alpha (0)=0,\dot{\alpha%
}(0)=1, \\
\ddot{\beta}(\tau )+\eta (\tau )\beta (\tau ) &=&0,\beta (0)=1,\dot{\beta}%
(0)=0,
\end{eqnarray*}%
and let 
\begin{equation*}
B(\tau )=%
\begin{pmatrix}
\beta (\tau ) & -\alpha (\tau ) \\ 
-\dot{\beta}(\tau ) & \dot{\alpha}(\tau )%
\end{pmatrix}%
.
\end{equation*}%
Assume $\beta $ is nonzero in some time interval $[0,T]$, then $\mu \left(
B\left( \tau \right) \right) f$ solves the Schr\"{o}dinger equation with
switchable quadratic trap:%
\begin{eqnarray}
i\partial _{\tau }u &=&\left( -\frac{1}{2}\partial _{y}^{2}+\eta (\tau )%
\frac{y^{2}}{2}\right) u\text{ in }\mathbb{R}\times \lbrack 0,T]
\label{eqn:the Schrodinger equation with switchable quadratic traps} \\
u(0,y) &=&f(y)\in L^{2}(\mathbb{R}).  \notag
\end{eqnarray}
\end{proposition}

\begin{proof}
We calculate%
\begin{eqnarray*}
\partial _{\tau }|_{\tau =0}\mu \left( B\left( \tau _{0}+\tau \right)
\right) f &=&\left( \partial _{\tau }|_{\tau =0}\mu \left( B\left( \tau
_{0}+\tau \right) \right) \right) f \\
&=&\left( \partial _{\tau }|_{\tau =0}\mu \left( B\left( \tau _{0}+\tau
\right) B^{-1}\left( \tau _{0}\right) B\left( \tau _{0}\right) \right)
\right) f \\
&=&\left( \partial _{\tau }|_{\tau =0}\mu \left( B\left( \tau _{0}+\tau
\right) B^{-1}\left( \tau _{0}\right) \right) \right) \mu \left( B\left(
\tau _{0}\right) \right) f \\
&=&d\mu (B^{\prime }(\tau _{0})B^{-1}\left( \tau _{0}\right) )\mu \left(
B\left( \tau _{0}\right) \right) f.
\end{eqnarray*}%
where%
\begin{eqnarray*}
B^{\prime }(\tau _{0})B^{-1}\left( \tau _{0}\right) &=&%
\begin{pmatrix}
\dot{\beta}(\tau _{0}) & -\dot{\alpha}(\tau _{0}) \\ 
-\ddot{\beta}(\tau _{0}) & \ddot{\alpha}(\tau _{0})%
\end{pmatrix}%
\begin{pmatrix}
\dot{\alpha}(\tau _{0}) & \alpha (\tau _{0}) \\ 
\dot{\beta}(\tau _{0}) & \beta (\tau _{0})%
\end{pmatrix}
\\
&=&%
\begin{pmatrix}
\dot{\beta}(\tau _{0}) & -\dot{\alpha}(\tau _{0}) \\ 
\eta (\tau _{0})\beta (\tau _{0}) & -\eta (\tau _{0})\alpha (\tau _{0})%
\end{pmatrix}%
\begin{pmatrix}
\dot{\alpha}(\tau _{0}) & \alpha (\tau _{0}) \\ 
\dot{\beta}(\tau _{0}) & \beta (\tau _{0})%
\end{pmatrix}
\\
&=&%
\begin{pmatrix}
0 & \dot{\beta}(\tau _{0})\alpha (\tau _{0})-\dot{\alpha}(\tau _{0})\beta
(\tau _{0}) \\ 
\eta (\tau _{0})\left( \dot{\alpha}(\tau _{0})\beta (\tau _{0})-\dot{\beta}%
(\tau _{0})\alpha (\tau _{0})\right) & 0%
\end{pmatrix}%
.
\end{eqnarray*}%
Notice that the Wronskian of $\alpha $ and $\beta $ is constant $1$ i.e.%
\begin{equation*}
\dot{\alpha}(\tau )\beta (\tau )-\alpha (\tau )\dot{\beta}(\tau )=1.
\end{equation*}%
So%
\begin{eqnarray*}
d\mu (B^{\prime }(\tau _{0})B^{-1}\left( \tau _{0}\right) ) &=&d\mu \left( 
\begin{pmatrix}
0 & -1 \\ 
\eta (\tau _{0}) & 0%
\end{pmatrix}%
\right) \\
&=&-\frac{i}{2}\left( -\partial _{y}^{2}+\eta (\tau _{0})y^{2}\right) .
\end{eqnarray*}%
In other words,%
\begin{equation*}
\partial _{\tau }\left( \mu \left( B\left( \tau \right) \right) f\right) =-%
\frac{i}{2}\left( -\partial _{y}^{2}+\eta (\tau )y^{2}\right) \left( \mu
\left( B\left( \tau \right) \right) f\right) .
\end{equation*}%
Before we end the proof, we remark that $\beta \neq 0$ is required for the
metaplectic representation to be well-defined.
\end{proof}

Through the LDU decomposition of the matrix $B$, we derive the generalized
lens transform. The LDU decomposition of the matrix $B$ is 
\begin{eqnarray*}
B(\tau ) &=&%
\begin{pmatrix}
\beta (\tau ) & -\alpha (\tau ) \\ 
-\dot{\beta}(\tau ) & \dot{\alpha}(\tau )%
\end{pmatrix}
\\
&=&%
\begin{pmatrix}
\beta (\tau ) & -\alpha (\tau ) \\ 
-\dot{\beta}(\tau ) & \alpha (\tau )\frac{\dot{\beta}(\tau )}{\beta (\tau )}+%
\frac{1}{\beta (\tau )}%
\end{pmatrix}
\\
&=&%
\begin{pmatrix}
1 & 0 \\ 
-\frac{\dot{\beta}(\tau )}{\beta (\tau )} & 1%
\end{pmatrix}%
\begin{pmatrix}
\beta (\tau ) & 0 \\ 
0 & \frac{1}{\beta (\tau )}%
\end{pmatrix}%
\begin{pmatrix}
1 & -\frac{\alpha (\tau )}{\beta (\tau )} \\ 
0 & 1%
\end{pmatrix}%
.
\end{eqnarray*}%
Hence we have%
\begin{equation}
\mu \left( B(\tau )\right) f=\mu \left( 
\begin{pmatrix}
1 & 0 \\ 
-\frac{\dot{\beta}(\tau )}{\beta (\tau )} & 1%
\end{pmatrix}%
\right) \mu \left( 
\begin{pmatrix}
\beta (\tau ) & 0 \\ 
0 & \frac{1}{\beta (\tau )}%
\end{pmatrix}%
\right) \mu \left( 
\begin{pmatrix}
1 & -\frac{\alpha (\tau )}{\beta (\tau )} \\ 
0 & 1%
\end{pmatrix}%
\right) f,  \label{equality:well-definedness}
\end{equation}%
where%
\begin{eqnarray*}
\mu \left( 
\begin{pmatrix}
1 & 0 \\ 
-\frac{\dot{\beta}(\tau )}{\beta (\tau )} & 1%
\end{pmatrix}%
\right) f(y) &=&e^{i\frac{\dot{\beta}(\tau )}{\beta (\tau )}\frac{y^{2}}{2}%
}f(y)\text{ by }\left( 4.25\right) \text{ in \cite{Folland}} \\
\mu \left( 
\begin{pmatrix}
\beta (\tau ) & 0 \\ 
0 & \frac{1}{\beta (\tau )}%
\end{pmatrix}%
\right) f(y) &=&\frac{1}{\left( \beta (\tau )\right) ^{\frac{1}{2}}}f(\frac{y%
}{\beta (\tau )})\text{ by }\left( 4.24\right) \text{ in \cite{Folland}} \\
\mu \left( 
\begin{pmatrix}
1 & -\frac{\alpha (\tau )}{\beta (\tau )} \\ 
0 & 1%
\end{pmatrix}%
\right) f(y) &=&e^{i\frac{\alpha (\tau )}{\beta (\tau )}\frac{\partial
_{y}^{2}}{2}}f\text{ by }\left( 4.54\right) \text{ in \cite{Folland}.}
\end{eqnarray*}%
Due to the definition of $\mu $, equality \ref{equality:well-definedness} in
fact holds up to a $"\pm "$ sign which depends on the time interval.
However, the LHS and the RHS of equality \ref{equality:well-definedness}
agree for sufficiently small $\tau $. By continuity, they must agree on the
time interval $\left[ 0,T\right] $ where $\beta \neq 0$. So we conclude the
following lemma concerning the generalized lens transform.

\begin{lemma}
\cite{Carles} Assume $\beta $ is nonzero in the time interval $[0,T]$, then
the solution of the Schr\"{o}dinger equation with switchable quadratic trap
(equation \ref{eqn:the Schrodinger equation with switchable quadratic traps}%
) in $[0,T]$ is given by%
\begin{equation*}
u(\tau ,y)=\frac{e^{i\frac{\dot{\beta}(\tau )}{\beta (\tau )}\frac{y^{2}}{2}}%
}{\left( \beta (\tau )\right) ^{\frac{1}{2}}}v(\frac{\alpha (\tau )}{\beta
(\tau )},\frac{y}{\beta (\tau )}),
\end{equation*}%
if $v(t,x)$ solves the free Schr\"{o}rdinger equation%
\begin{eqnarray*}
i\partial _{t}v &=&-\frac{1}{2}\partial _{x}^{2}v\text{ in }\mathbb{R}^{1+1}
\\
v(0,x) &=&f(x)\in L^{2}(\mathbb{R}).
\end{eqnarray*}
\end{lemma}

The anisotropic case, Lemma \ref{Lemma:GLensTransform}, follows from the
above lemma.

\subsection{Proof of Lemma \protect\ref{Lemma:Monentum} / Evolution of
Momentum}

Using the metaplectic representation, we can also compute the evolution of
momentum and position.

\begin{lemma}
The evolution of momentum and position is given by%
\begin{eqnarray*}
P(\tau ) &=&\mu \left( B(\tau )\right) \circ \left( -i\partial _{y}\right)
\circ \left( \mu \left( B(\tau )\right) \right) ^{-1}=-i\beta (\tau
)\partial _{y}-\dot{\beta}(\tau )y \\
Y(\tau ) &=&\mu \left( B(\tau )\right) \circ y\circ \left( \mu \left( B(\tau
)\right) \right) ^{-1}=i\alpha (\tau )\partial _{y}+\dot{\alpha}(\tau )y.
\end{eqnarray*}
\end{lemma}

\begin{proof}
Let us only compute the momentum, position can be obtained similarly.%
\begin{eqnarray*}
\mu \left( B(\tau )\right) \left( -i\partial _{y}\right) \left( \mu \left(
B(\tau )\right) \right) ^{-1} &=&\mu \left( B(\tau )\right) 
\begin{pmatrix}
1 & 0%
\end{pmatrix}%
\begin{pmatrix}
-i\partial _{y} \\ 
y%
\end{pmatrix}%
\left( \mu \left( B(\tau )\right) \right) ^{-1} \\
&=&%
\begin{pmatrix}
1 & 0%
\end{pmatrix}%
\left( B(\tau )\right) ^{T}%
\begin{pmatrix}
-i\partial _{y} \\ 
y%
\end{pmatrix}%
\text{ (Theorem 2.15 in \cite{Folland})} \\
&=&%
\begin{pmatrix}
1 & 0%
\end{pmatrix}%
\begin{pmatrix}
\beta (\tau ) & -\dot{\beta}(\tau ) \\ 
-\alpha (\tau ) & \dot{\alpha}(\tau )%
\end{pmatrix}%
\begin{pmatrix}
-i\partial _{y} \\ 
y%
\end{pmatrix}
\\
&=&-i\beta (\tau )\partial _{y}-\dot{\beta}(\tau )y
\end{eqnarray*}
\end{proof}

\begin{remark}
We select $-i\partial _{y}$ to be the momentum to match the canonical
commutation relations in Folland \cite{Folland} which is 
\begin{equation*}
\left[ -i\partial _{y},y\right] =-iI.
\end{equation*}
\end{remark}

The above lemma reproduces the following result in Carles \cite{Carles}.

\begin{lemma}
\cite{Carles} The operators $P(\tau )$ and $Y(\tau )$ commute with the
linear operator 
\begin{equation*}
i\partial _{\tau }+\frac{1}{2}\partial _{y}^{2}-\eta (\tau )\frac{y^{2}}{2}
\end{equation*}%
Moreover, 
\begin{eqnarray*}
P(\tau )U(\tau ;s) &=&U(\tau ;s)P(s) \\
Y(\tau )U(\tau ;s) &=&U(\tau ;s)Y(s)
\end{eqnarray*}%
if we let $U_{y}(\tau ;s)$ be the solution operator of 
\begin{eqnarray*}
i\partial _{\tau }u &=&\left( -\frac{1}{2}\partial _{y}^{2}+\eta (\tau )%
\frac{y^{2}}{2}\right) u\text{ in }\mathbb{R}^{1+1} \\
u(s,y) &=&u_{s}(y)\in L^{2}(\mathbb{R}),
\end{eqnarray*}%
or in other words%
\begin{equation*}
U_{y}(\tau ;s)=\mu \left( B(\tau )\right) \mu \left( B(s)\right) ^{-1}.
\end{equation*}
\end{lemma}

Thence we have shown Lemma \ref{Lemma:Monentum}.

\end{document}